\documentclass[11pt]{article}
\usepackage[pdftex,margin=1in]{geometry}

\usepackage{amsmath,amssymb,amsfonts,amsthm,bm}

\usepackage{graphicx, color}
\usepackage{epstopdf,dsfont}
\usepackage{hyperref}

\usepackage{subfigure}% Support for small, `sub' figures and tables

\usepackage{natbib}% Citation support using natbib.sty
\bibpunct[, ]{(}{)}{;}{a}{}{,}% Citation support using natbib.sty
\usepackage{pdfpages}

\theoremstyle{plain}% Theorem-like structures provided by amsthm.sty
\newtheorem{theorem}{Theorem}[section]

\newtheorem{corollary}[theorem]{Corollary}
\newtheorem{proposition}[theorem]{Proposition}

\theoremstyle{definition}

\theoremstyle{remark}
\newtheorem{remark}{Remark}

\DeclareMathOperator*{\arginf}{arg\,inf}
\DeclareMathOperator*{\diag}{diag}

\newcommand{\cD}{\mathcal D}

\newcommand{\cX}{\mathcal X}

\newcommand{\cN}{\mathcal N}

\newcommand{\bI}{\mathbf I}
\newcommand{\bK}{\mathbf K}

\newcommand{\E}{\mathbb{E}}
\newcommand{\R}{\mathbb{R}}

\newcommand{\bx}{\mathbf{x}}
\newcommand{\vy}{\vec{\mathbf{y}}}

\newcommand{\vb}{\bm{\beta}}

\newcommand{\vm}{\vec{\mathbf{m}}}

\author{Mike Ludkovski and Yuri Saporito}
\title{KrigHedge: Gaussian Process Surrogates for Delta Hedging}

\begin{document}

\maketitle

\begin{abstract}
  We investigate a machine learning approach to option Greeks approximation based on Gaussian Process (GP) surrogates. Our motivation is to implement Delta hedging in cases where direct computation is expensive, such as in local volatility models, or can only ever be done approximately. The proposed method takes in noisily observed option prices, fits a nonparametric input-output map and then analytically differentiates the latter to obtain the various price sensitivities. Thus, a single surrogate yields multiple self-consistent Greeks.  We provide a detailed analysis of numerous aspects of GP surrogates, including choice of kernel family, simulation design, choice of trend function and impact of noise.
   We moreover connect the quality of the Delta approximation to the resulting discrete-time hedging loss. Results are illustrated with two extensive case studies that consider estimation of Delta, Theta and Gamma and benchmark approximation quality and uncertainty quantification using a variety of statistical metrics. Among our key take-aways are the recommendation to use Mat{\'e}rn kernels, the benefit of including virtual training points to capture boundary conditions, and the significant loss of fidelity when training on stock-path-based datasets.
\end{abstract}

\section{Introduction}

Fundamentally, hedging is about learning the sensitivities of the contingent claim to evolving market factors. For example, Delta hedging manages risk by controlling for the sensitivity of the financial derivative to the underlying spot price. Theta manages risk by controlling for the sensitivity of the financial derivative to the passing of time, and so on. Thus, successful hedging strategies depend on accurately learning such sensitivities. Unfortunately the related Greeks are rarely available analytically, motivating the large extant literature %\todo{CITE: Fournie, Benhamou, Saporito, Chen \& Glasserman, Giles}
\citep{capriotti2017aad,fu2012estimating,jazaerli2017functional,ruf2019neural,ruf2020hedging} on Greek approximation and computation. The goal of this article is to contribute to this enterprise by investigating a novel tie-in between machine learning and hedging. The idea is to develop a non-parametric method that does {not} require working with any particular stochastic model class---all that is needed is the data source (or a black-box simulator)  generating \emph{approximate option prices}. The training dataset is used to fit a data-driven input-output mapping and evaluate the respective price sensitivity. Specifically, we propose to use Gaussian Process (GP) surrogates to capture the functional relationship between derivative contract price and relevant model parameters, and then to analytically differentiate the fitted functional approximator to extract the Delta or other desired Greek.

Our specific implementation  brings several advantages over competing methods. First, GPs can handle both interpolation and smoothing tasks, i.e.~one may treat training inputs as being exact or noisy. Therefore, GP surrogates can be applied across the contexts of (a) speeding up Greek computations when a few exact data samples are available (model calibration), of (b) utilizing approximate Monte-Carlo-based samples, and of (c) fitting to real-life data. Second, GPs are well-suited to arbitrary training sets and so naturally accommodate historical data that is highly non-uniform in the price dimension (namely based on a historical trajectory of the underlying). Third, GPs offer \emph{uncertainty quantification} so rather than providing a best-estimate of the desired Greek, GPs further supply a state-dependent confidence interval around that estimate. This interval is crucial for hedging purposes, since it indicates how strict one ought to be in matching the target Greek.  %Fourth, the GPs can be customized to handle constraints, in particular model-free no-arbitrage restrictions that must be satisfied by the underlying input-output statistical model.
Fourth, GPs interact well with dynamic training, i.e.~settings where the training sets change over time.

Differently to our approach presented here, GP regression have been applied to other financial mathematics' problems. For instance, in \cite{MadanSchoutens18}, the authors considered GPs to speed-up pricing of derivatives contracts (including exotic ones) within reasonable reduction of accuracy. Additionally, they also applied GP regression to interpolate implied volatility surfaces and use it for backtesting an option strategy.

Considering a different application, \cite{DixonCrepey19} applied multiple-output GPs to speed-up mark-to-market of derivative portfolios in the context of credit valuation adjustment (CVA). Moreover, the authors also use single-output GP to exemplify the learning of the pricing formula of financial models as the Heston model. Similarly to our approach, they also mention that GPs provides analytic expression for sensitivities of derivative prices. However, differently from our approach, they consider only the Black--Scholes model where the GP regression is trained using the Black--Scholes formula and, although Greek approximation is considered, the implications to the hedging problem are not studied. Furthermore, \cite{DixonCrepey21} used GPs to fit no-arbitrage constrained option price surfaces given empirical quotes, while \cite{GoudenegeZanette20} applied them for value function approximation of American options.

Existing literature on numerical Greeks approximation is generally split between the noiseless setting (known as curve-fitting or interpolation) and the noisy case (statistical regression). For interpolation, the state of the art are the Chebyshev polynomials recently studied in the series of works \cite{GlauGass18,GlauMadan19,GlauMahlstedt19}.  For regression many of the best performing methods, such as random forests, are not even differentiable so do not necessarily possess gradients. In contrast, GPs gracefully unify in a single framework both the noiseless and noisy settings.

Within this landscape our contribution is to provide a detailed analysis of GP surrogates for Greek approximation and Delta hedging. To this end, we investigate the role and relative importance of various surrogate ingredients, such as kernel family, shape of experimental design, training data size, and propose several modifications that target the financial application. Moreover, we assess the performance of our Greek approximators both from the statistical perspective, as well as from the trader's perspective in terms of the resulting hedging error. In particular, in Proposition 1 we connect the quality of the Delta approximation with the resulting hedging loss, providing insights into how errors in estimating the Greeks translate into the hedging P\&L.

The rest of the paper is organized as follows. Section \ref{sec:model} explains our approach of approximating price sensitivities using GP surrogates. Section \ref{sec:bs}
presents numerical experiments in the classical Black--Scholes model, while Section \ref{sec:lv} does the same for a local volatility model where ground truth is no longer immediately available. Section \ref{sec:discuss} discusses our findings and outlines future research directions.

%\subsubsection*{Delta Hedging}

\section{Modeling the Option Price Surface}\label{sec:model}
To fix ideas, consider hedging of a European Call contract. The European Call has a given strike $K$ and maturity $T$ and is written on underlying $(S_t)_{t \geq 0}$.
 The respective  no-arbitrage option price is given by (we also use $P$ to denote a generic contract price function)
\begin{align}\label{eq:call}
  P(t, S) := \E^Q \left[ e^{-r (T-t)} (S_T - K)_+ |S_t = S \right],
\end{align}
where we emphasize the dependence on the calendar time $t$ and the current spot price $S$. Above $Q$ is a pricing martingale measure, kept generic for now.
Any other European-style financial contract can be similarly considered; we do not make any direct use of the specific form (or smoothness) of the Call payoff in \eqref{eq:call} henceforth.

As a canonical example of hedging, we are interested in finding the Call Delta $$\Delta(t,S) := \partial P(t, S)/\partial S,$$
for arbitrary $(t,S)$.
 In the most classical setting (such as the Black--Scholes model), one has an analytical formula for $(t,S) \mapsto P(t,S)$ and can then simply differentiate the latter to obtain the Delta. We rather consider the more common situation where $P(t,S)$ is not directly known. Instead, we are provided a training set $\cD = \{ (t_i, S_i, Y_i) : i=1,\ldots, N\}$, where $Y_i \simeq P(t_i, S_i)$, and have to use this data to \emph{infer} or learn $(t,S) \mapsto \widehat{\Delta}(t,S)$. This problem is motivated by the situation where a pricing model is available but it is computationally expensive to directly compute $P(t,S)$ each time the option price is needed and so only a sample of such computations is provided. We distinguish two sub-cases:
%\commy{should we change I.a. to a.?}

\textbf{(a)} Computing $P(t,S)$ exactly is possible, but is challenging/time-consuming. For example, it might necessitate solving a partial differential equation. Then $\cD$ is a collection of inputs where $Y_i = P(t_i,S_i)$ was evaluated exactly and the goal is to obtain a cheap \emph{representation} of the map $(t,S) \mapsto \Delta(t,S)$ by extrapolating the exact $Y_i$'s.

\textbf{(b)} Option prices are evaluated through a Monte Carlo engine. For a given $(t_i, S_i)$, the modeler has access to an empirical average $Y_i$ of $\check{N}$ Monte Carlo samples, with precision being on the  order of $\mathcal{O}(\check{N}^{-1/2})$. For finite $\check{N}$, $Y_i$ is a \emph{noisy} estimate of $P(t_i, S_i)$. The training set $\cD$ is then a collection of such noisy samples that need to be smoothed, interpolated and differentiated to learn the map $(t,S) \mapsto \Delta(t,S)$.

We note that because the training data is generated by the modeler, there is the related question of \emph{experimental design}, i.e.~how to choose  $\cD$ wisely to maximize computational efficiency that we will also explore.

\subsection{Surrogate Gradients}
In both Setting (a) and Setting (b) above, our aim is to provide an estimate of $\Delta(t,S)$ for arbitrary $(t,S) \in \cD'$ in some test set $\cD'$. This could include \emph{in-sample} predictions, i.e.~for $(t,S) \in \cD$ (for example obtaining Delta at same inputs used for training), or \emph{out-of-sample} predictions, including extrapolation $(t,S)$ outside the convex hull of $\cD$ which would be the case if training is confined to the past $t < T_0$ and we want to Delta hedge in the future, $t > T_0$.
Analogously to needing $\Delta(t,S)$ in order to \emph{hedge} the respective risk of the underlying price moves, the trader is also interested in other Greeks. Two examples that we will also consider below include the Theta $\Theta(t,S)$---sensitivity to $t$, and the Gamma $\Gamma(t,S)$---the second derivative of $P(t,S)$ with respect to $S$.

We emphasize that while the training set contains information about option price $P(t_i, S_i)$, our goal is to learn the price \emph{sensitivities}. We tackle this issue by using the intermediate step of first fitting a statistically-driven non-parametric mapping $(t,S) \mapsto \widehat{P}(t,S)$ called a \emph{surrogate} or a metamodel. We then set $\widehat{\Delta}(t,S) := \partial \widehat{P}(t,S)/\partial S$. A key idea is that the second step of taking derivatives is done analytically, even though $\widehat{P}(t,S)$ is non-parametric. On the one hand, this strategy reduces the error in $\widehat{\Delta}$ since only a single metamodeling approximation is needed and the differentiation is exact. On the other hand, it offers precise \emph{uncertainty quantification}, offering an in-model assessment of the accuracy of $\widehat{\Delta}$ by rigorously propagating the underlying uncertainty in $\widehat{P}$. In particular, the method provides credible bands around $\widehat{\Delta}$, giving the end-user a clear guidance on how well is the model learning the Greek. This information is critical for trading purposes, in particular in the context of no-transaction regions under transaction cost regimes, see for instance \cite{WhalleyWilmott1997}.

%\begin{rem}
Our data-driven approach is broadly known as curve-fitting. In general, parametric curve-fitting via constructing a surrogate $(t,S) \mapsto \widehat{P}(t,S)$  (e.g.~via a spline-based $\widehat{P}$) and then differentiating it is known to lead to highly unstable estimates for $\widehat{\Delta}$ and other gradients. This is because the typical $L^2$ criterion that is driving the fitting of $\widehat{P}(t,S)$ is completely unaware of the subsequent plan to compute gradients. As a result, differentiating the typical regression fit can lead to nonsensical gradient estimates, see e.g.~\cite{Jain15}. The machine learning folklore (e.g. in the context of vast Bayesian optimization literature) suggests that GPs, which can be understood as a type of kernel regression with smoothness penalties, are often able to mitigate this concern.
%\end{rem}

\subsection{Gaussian Process Regression}

We temporarily restrict attention to a single-factor model,
  viewing $P(t, S)$ as a 2D surface in the two coordinates of (`time') and (`stock'), encoded as $\bx = (x_1, x_2) \equiv (t, S)$, i.e.~$\bx \mapsto P(\bx)$ is a function in $\R^d$ with $d=2$. %We furthermore imagine a calibration setting, namely that we are provided some observations of $P(\tau_i, s_i)$ for some input dataset $\cD = \{ (\tau_i, s_i) : i=1,\ldots, N\}$ and we wish to use that information to estimate $\Delta(\tau, S)$ or $\Theta(\tau, S)$ for arbitrary $x$.
 We treat the two coordinates in a symmetric manner for fitting purposes. As a result, the Delta is viewed as one specific instance of the gradient of $P$. Multi-factor models (with fully observed factors) would simply correspond to working in higher $d>2$.

 The curve fitting for $\widehat{P}$ is carried out using {a regularized} $L^2$ regression framework, namely finding the best approximator in a given normed space $\mathcal{H}$ conditional on the training set $\cD$ of size $N$:
 \begin{align}\label{eq:fn-approx}
   \widehat{P} = \arginf_{ P \in \mathcal{H}} \sum_{i=1}^N| P(\bx^i) - Y^i|^2 + \| P \|_{\mathcal{H}}.
 \end{align}
{ The last term acts as a regularizer, balancing quality of fit and the prior likelihood of the approximator.}

We propose to use Gaussian Process regression (GPR) for the purpose of learning the price surface
$\widehat{P}(\bx)$ based on the observation model
\begin{align}\label{eq:noise}
Y(\bx) = P(\bx) + \epsilon(\bx).
\end{align}
Above we distinguish between the \emph{true} price map $P(\bx)$ and the observed price $Y(\bx)$ which may/may not be the same.
Gaussian process regression is a flexible non-parametric regression method \citep{rasmussen2006gaussian} that views the map $\bx \rightarrow P(\bx) $ as a realization of a Gaussian random field so that (in the abstract metamodel probability space, which is independent of the probabilistic structure present in asset stochastic dynamics) any finite collection of $\{P(\bx), \bx \in \cX \}$, is multivariate Gaussian. For any $n \ge 1$ design sites $\{\bx^i\}_{i=1}^{n}$, GPR posits that
\begin{align*}
(P(\bx^1),\ldots,P(\bx^n)) \sim \cN(\vm_n , \bK_n)
\end{align*}
with mean vector $\vm_n := [ m(\bx^1; \vb),\ldots, m(\bx^{n}; \vb)]$
and $n\times n$ covariance matrix $\bK_n$ comprised of $\kappa(\bx^i,\bx^{i'}; \vb), \text{ for } 1 \leq i,i' \leq n $. The vector $\vb$ represents all the hyperparameters for this model. The role of $m(\cdot)$ is to capture the known trends in the response, and the role of $\kappa(\cdot, \cdot)$ is to capture the spatial dependence structure in $\bx \mapsto P(\bx)$.

Given the training dataset $\cD = \{\bx^i,Y^i \}_{i=1}^{N}$, GPR infers the posterior of $P(\cdot)$ by assuming an observation model \eqref{eq:noise} with a Gaussian noise term $\epsilon(\bx) \sim \cN(0,\sigma^2_{\epsilon})$. Conditioning equations for multivariate normal vectors imply that the posterior predictive distribution $P({\bx}_*)|\{\bx^i,Y^i \}_{i=1}^{N}$ at any arbitrary input ${\bx}_*$ is also Gaussian with the posterior mean $m_*({\bx}_*)$ that is the proposed estimator of $P({\bx}_*)$:
\begin{align}
\label{eq:gp_mean}
&m_*(\bx_*) := m(\bx_*)+ K^T(\bK + \sigma^2_\epsilon \bI)^{-1}(\vy-\vm ) = \mathbb{E}\Big[P(\bx_*)\big|\vec{\bx}, \vy\Big]; \\ \notag
&\text{where}\quad \vec{\bx} = [\bx^1,\ldots,\bx^{N} ]^T,  \ \  \vy = [y^1,\ldots,y^{N} ]^T , \\ \notag
&\phantom{\text{where}\quad} K^T = [\kappa(\bx_*, \bx^1; \vb), \ldots, \kappa(\bx_*,\bx^{N};  \vb)],\\ \notag
&\phantom{\text{where}\quad} \vm= [ m(\bx^1; \vb),\ldots,  m(\bx^{N}; \vb)],
\end{align}
and $\bK$ is $N \times N$ covariance matrix described through the kernel function $\kappa(\cdot,\cdot;\vb)$. Henceforth we think of $m_*(\cdot) \equiv \widehat{P}(\cdot)$ as a (smooth) function, even though it is only defined pointwise via \eqref{eq:gp_mean}.

The posterior covariance is
\begin{align}\label{eq:gp_cov}
\mbox{Cov}( P(\bx^1_*), P(\bx^2_*))  & = {\kappa(\bx^1_*,\bx^2_*)}-K^T_1{[\bK+\sigma^2_\epsilon \bI]}^{-1}K_2,
\end{align}
where $K_i = [\kappa(\bx_*^i, \bx^1; \vb), \ldots, \kappa(\bx^i_*,\bx^{N};  \vb)]$ for $i=1,2$.

The interpretation is that $\bx \mapsto m_*(\bx)$ is the ``most likely'' input-output map that is consistent with the training dataset $\cD$ and $\mbox{Var}(P(\bx))$ is the model uncertainty capturing the range of other potential input-output maps that could also be consistent (but less likely) with $\cD$.

\subsection{Specifying a GP Surrogate}\label{sec:gp-surrogates}
Returning to the curve-fitting perspective, the optimization in \eqref{eq:fn-approx} is available in closed-form through the kriging equations \eqref{eq:gp_mean} and GP fitting in fact corresponds to selecting an appropriate function space $\mathcal{H} \equiv \mathcal{H}_\vartheta$ by optimizing the hyper-parameters $\vartheta$. This is done in a hierarchical manner, first fixing a kernel family and then using maximum likelihood optimization to select $\vartheta$.

The GP kernel $\kappa(\bx,\bx')$ controls the smoothness (in the sense of differentiability) of $\widehat{P}$ and hence the roughness of its gradient.
 A popular choice for  $\kappa(\cdot,\cdot)$ is the (anisotropic) squared exponential (SE) family, parametrized by the lengthscales $\{ \ell_{\mathrm{len},k} \}_{k=1}^d $ and the process variance $\sigma_p^2$ :
\begin{equation}
 \kappa_{SE}(\bx,\bx'): = \sigma_p^2\exp{\Big(- \sum_{k=1}^d \frac{(x_{k} - x'_{k} )^2  }{2 \ell^2_{\mathrm{len},k}}\Big)}.\label{eqn:se-kernel}
\end{equation}
%This is the kernel used to produce Figure~\ref{fig:bs}.
The SE kernel~\eqref{eqn:se-kernel} yields infinitely differentiable fits $m_*(\cdot)$. Besides squared exponential kernel described above, other popular kernels include Mat\'ern-3/2 (henceforth, M32) and Mat\'ern-5/2 (M52)~\citep{DiceKriging}:
\begin{align}\label{eq:matern52}
& \kappa_{M52}(\bx,\bx') := \sigma_p^2\prod_{k=1}^d \left(1 + \frac{\sqrt{5}}{\ell_{\mathrm{len},k}}|x_k - x'_{k}| + \frac{5}{3\ell^2_{\mathrm{len},k}}(x_k - x'_{k})^2\right) e^{-\frac{\sqrt{5}}{\ell_{\mathrm{len},k}}|x_k - x'_{k}|},\\ \label{eq:matern32}
& \kappa_{M32}(\bx,\bx') := \sigma_p^2\prod_{k=1}^d \left(1 + \frac{\sqrt{3}}{\ell_{\mathrm{len},k}}|x_k - x'_{k}|\right) e^{-\frac{\sqrt{3}}{\ell_{\mathrm{len},k}}|x_k - x'_{k}|}.
\end{align}
{A Mat\'ern kernel of order $k+1/2$ yields approximators that are in $C^{k}$. Thus Mat\'ern-3/2 fits are in $C^1$ and Mat\'ern-5/2 fits are in $C^2$.}

The mean function is often assumed to be constant $m(\bx;  \vb) = \beta_0$ or described using a linear model $m(\bx;  \vb) = \sum_{k=1}^{K} \beta_k \phi(\bx)$ with $\phi(\cdot)$ representing a polynomial basis. The mean function drives the estimates during extrapolation (far out-of-sample) and also can strongly impact the gradient. For example, incorporating a convex quadratic prior mean compared to a flat linear prior mean modifies the curvature/lengthscales of $\widehat{P}$ and therefore affects the estimated Greek. 
The overall set of the hyperparameters for the GP surrogate is $\vb := ( \{\beta_{k} \}_{k=1}^{K}, \{\ell_{\mathrm{len},k} \}_{k=1}^d,  \sigma^2_{p}, \sigma^2_{\epsilon})$.

Typically one estimates $\vb$ by maximizing the log-likelihood function using the dataset $\{\bx^i,Y^i \}_{i=1}^{N}$. 

\subsection{Obtaining the Greek}\label{sec:greek}

Given a fitted GP model $f_* \sim GP( m_*, K_*)$, its gradient with respect to the coordinate $x_j$ forms another GP, $D \sim GP( g_*, K_g)$. The respective mean at input $\bx_*$ and covariance of $D$ at $\bx_*,\bx_*'$  are specified by
\begin{align}\label{eq:grad-mean}
 &g_*(\bx_*) := \frac{\partial m_*}{\partial x_j}(\bx_*) = \frac{\partial m}{\partial x_j}(\bx_*) + \frac{\partial \kappa}{\partial x_j}(\bx_*, \vec{\bx}) (\bK + \sigma^2_\epsilon \bI)^{-1}(\vy-\vm),\\ \label{eq:grad-sd}
&K_g(\bx_*,\bx_*') = \frac{\partial^2 K_*}{\partial x_j \partial x_j'} = \frac{\partial^2 \kappa}{\partial x_j \partial x_j'}(\bx_*,\bx_*') - \frac{\partial \kappa}{\partial x_j}(\bx_*, \vec{\bx}) (\bK + \sigma^2_\epsilon \bI)^{-1} \frac{\partial \kappa}{\partial x_j'}(\vec{\bx}, \bx_*').
\end{align}

Thus, the gradient estimator is $g_*(\bx_*)$ in \eqref{eq:grad-mean} which can be interpreted as formally differentiating the expression for $m_*(\cdot)$ with respect to $x_j$. Remarkably, the same procedure yields the posterior variance $V_g(\bx_*) = K_g(\bx_*, \bx_*)$ of $g_*(\bx_*)$ in \eqref{eq:grad-sd} and therefore we obtain \emph{analytically} the credible bands around $g_*(\bx_*)$. Namely, the credible band for $\frac{\partial P}{\partial x_j}(\bx_*)$ is
\begin{align}\label{eq:ci}
  CI_{\alpha}(\bx_*) := \left[ g_*(\bx_*) - z_\alpha \sqrt{V_g(\bx_*)}, g_*(\bx_*) + z_\alpha \sqrt{V_g(\bx_*)} \right]
\end{align}
where $z_\alpha$ is the desired quantile of the standard normal distribution, e.g.~$z_{0.95} = 1.96$ to obtain 95\% CI. The upshot is that once a GP surrogate is fit to option prices, obtaining Greek estimates and their model-based uncertainty reduces to evaluating the formulas \eqref{eq:grad-mean}-\eqref{eq:grad-sd}.

As an example of such calculations we present the analytic expression for $g_*(\cdot)$ for the three most common kernels discussed in Section~\ref{sec:gp-surrogates}. While these computations are not new, we also could not find any handy reference for them in the literature.
For the SE kernel \eqref{eqn:se-kernel} we have:
\begin{align}
&\frac{\partial \kappa_{SE}}{\partial x_j}(\bx,\bx') = 2\frac{x_j' - x_j  }{\ell_{\mathrm{len},j}} \kappa_{SE}(\bx,\bx'),\\
& V_g(\bx_*) = \frac{2}{\ell_{\mathrm{len},j}^2} \sigma^2_p - \frac{\partial \kappa_{SE}}{\partial x_j}(\bx_*, \bx) (\bK_{SE} + \sigma^2_\epsilon \bI)^{-1} \frac{\partial \kappa_{SE}}{\partial x_j}(\bx, \bx_*).
%&\frac{\partial^2 \kappa_{SE}}{\partial x_j \partial x_j'}(\bx,\bx') = 2\left(\frac{1}{\ell_{\mathrm{len},j}} - 2 \frac{(x_j' - x_j)^2 }{\ell_{\mathrm{len},j}^2} \right) \kappa_{SE}(\bx,\bx').
\end{align}

For the Mat\'ern-5/2 kernel~\eqref{eq:matern52}, we find
\begin{align}
&\frac{\partial \kappa_{M52}}{\partial x_j}(\bx,\bx') = \left(\frac{-\frac{5}{3\ell_{\mathrm{len},j}^2}(x_j - x'_{j}) - \frac{5^{3/2}}{3\ell_{\mathrm{len},j}^3}(x_j - x'_{j})|x_j - x'_{j}|}{1 + \frac{\sqrt{5}}{\ell_{\mathrm{len},j}}|x_j - x'_{j}| + \frac{5}{3\ell^2_{\mathrm{len},j}}(x_j - x'_{j})^2} \right) \kappa_{M52}(\bx,\bx'),\\
& V_g(\bx_*) = -\frac{5}{3\ell_{\mathrm{len},j}^2} \sigma^2_p - \frac{\partial \kappa_{M52}}{\partial x_j}(\bx_*, \bx) (\bK_{M52} + \sigma^2_\epsilon \bI)^{-1} \frac{\partial \kappa_{M52}}{\partial x_j}(\bx, \bx_*),
\end{align}
and for the Mat\'ern-3/2 kernel~\eqref{eq:matern32}:
\begin{align}
&\frac{\partial \kappa_{M32}}{\partial x_j}(\bx,\bx') = \left(\frac{-\frac{3}{\ell_{\mathrm{len},j}^2}(x_j - x'_{j})}{1 + \frac{\sqrt{3}}{\ell_{\mathrm{len},j}}|x_j - x'_{j}| } \right) \kappa_{M32}(\bx,\bx'),\\
& V_g(\bx_*) = -\frac{3}{\ell_{\mathrm{len},j}^2} \sigma^2_p - \frac{\partial \kappa_{M32}}{\partial x_j}(\bx_*, \bx) (\bK_{M32} + \sigma^2_\epsilon \bI)^{-1} \frac{\partial \kappa_{M32}}{\partial x_j}(\bx, \bx_*).
\end{align}

We emphasize that the above formulas work both for the Delta $\partial P/\partial x_2$ and the Theta $-\partial P/\partial x_1$, with the GP model yielding analytic estimates of all gradients simultaneously, without the need for any additional training or computation.

\begin{remark}
The underlying structure is that differentiation is a linear operator that algebraically ``commutes'' with the Gaussian distributions defining a GP model. Consequently, one may iterate (by applying the chain rule further  on $\kappa(\cdot, \cdot)$ and its derivatives, provided they exist) to obtain analytic expressions for the mean and covariance of higher-order partial derivatives of $f$, yielding
 second-order and higher option sensitivities, for example the Gamma. Instead of doing so, we implemented a finite-difference estimator for $\Gamma(t,S)$:
  \begin{align}\label{eq:fd}
    \widehat{\Gamma}^{fd}(t,S; \delta) := \frac{ \widehat{P}(t,S+\delta) - 2\widehat{P}(t,S) + \widehat{P}(t,S-\delta)}{\delta^2}
  \end{align}
  for a discretization parameter $\delta > 0$. By predicting the GP model on the triplet of sites $\{ (t,S-\delta), (t,S), (t,S+\delta)\}$ we obtain the predictive covariance matrix and can use that to compute the variance of $\widehat{\Gamma}^{fd}(t,S; \delta)$ (which is a linear combination of the respective three $\widehat{P}$ values).  Note that the Mat\'ern-3/2 kernel is not twice differentiable, so formally there is no second sensitivity and we expect numeric instability in applying \eqref{eq:fd} to a M32-based model.
\end{remark}

\subsection{Illustration}

Figure \ref{fig:bs} shows the GP-based $\widehat{\Delta}(t,\cdot)$ for the case of a Call option $P(t,S)$ within a Black--Scholes model with constant coefficients $r=0.04,\sigma=0.22,T=0.4,K=50$, parametrized by time-to-maturity $\tau$ and spot price $S$. The model is trained on a two-dimensional $10 \times 10$ grid (so that $N=100$) $S^i \in \{32, 36, \ldots, 68\}, \tau^i \in \{0.04, 0.08, \ldots, 0.4\}$, using for inputs the exact $P(\tau^i, S^i), i=1,\ldots 100$ available via the Black--Scholes formula. We then display two 1-D slices of the resulting estimate of the Delta $\Delta(\tau,S)$ as a function of spot $S$, keeping time-to-maturity $\tau$ fixed. In the left panel we look at $\tau = 0.5$ which is an extrapolation relative to the training set, maturity being longer than $\bar{\tau} = 0.4$. In the right panel we use $\tau = 0.2$ which is one of the training times-to-maturity, and corresponds to in-sample interpolation. Note that with GPs the two computations are implemented completely identically. In Figure \ref{fig:bs} we compare $\widehat{\Delta}(\tau, \cdot)$ to the exact ground truth $\Delta(\tau, \cdot)$ and also display the corresponding 95\% posterior credible bands, cf.~\eqref{eq:ci} below. We observe that the GP fit is excellent, being indistinguishable from the ground-truth for most of the test locations. While the goodness-of-fit is relatively good in the middle, towards the edges we have numerical artifacts, such as $\widehat{\Delta}$ being outside the interval $[0,1]$ or not being increasing in $S$.

\begin{figure}[ht]
	\begin{tabular}{cc}
	\begin{minipage}[t]{0.48\linewidth}
		\begin{center}
			\includegraphics[width=1\textwidth,height=2.4in,trim=0.2in 0.1in 0in 0.5in]{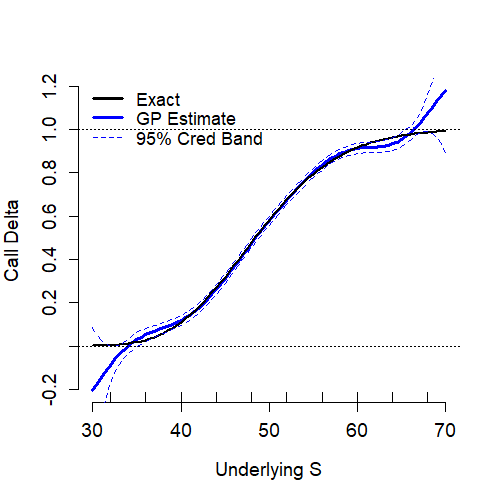}
		\end{center}
	\end{minipage} &
	\begin{minipage}[t]{0.48\linewidth}
		\begin{center}
			\includegraphics[width=1\textwidth,height=2.4in,trim=0.2in 0.1in 0in 0.5in]{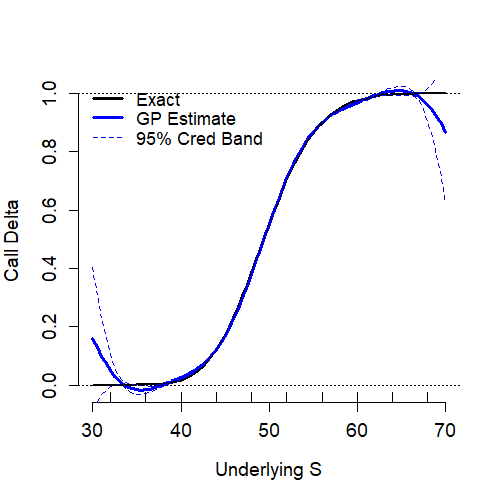}
		\end{center}
	\end{minipage} \\
$\tau = 0.5$ & $\tau = 0.2$
 \end{tabular}
	\caption{Estimated Delta $\widehat{\Delta}$ for a Black--Scholes Call. Left panel extrapolates for longer maturity $\tau=0.5$; right panel is an in-sample test set $\tau=0.2$. We also show the 95\% credible bands from \eqref{eq:ci}. %Exact GP derivative in blue and a FD approximation~\eqref{eq:fd} in purple. The two agree closely except at the extremes.
	\label{fig:bs}}
\end{figure}

As expected, the credible bands on the right (the interpolation case, where credible bands are almost invisible in the middle) are narrower than on the left (extrapolation). For example at $S=55$ and $t=0.2$ we have $\widehat{\Delta}(t,S) = 0.8681$ with a credible band of $[0.8635, 0.8727]$ (the true Delta actually being $0.8642$) while at same $S=55$ and $\tau=0.5$ we have $\widehat{\Delta}(t,S) = 0.7997$ (ground truth being $0.7936$) and a credible band of $[0.7779, 0.8215]$ more than 4 times wider. In other words, this particular GP surrogate is able to estimate Delta up to $\pm 0.004$ in the middle of the training set, but only up to $\pm 0.022$ when asked to extrapolate for longer maturity. This reflects the key feature of GPs that the fitted model is ``self-aware'' and more confident in its estimate in regions that are close to training locations. The latter notion of closeness is algebraically reflected in the fitted covariance kernel $\kappa$, specifically its lengthscale $\ell_1$. Figure~\ref{fig:bs} moreover visualizes the dependence of uncertainty quantification on $S$: in the middle of the training set $S \in [40,60]$, the bands are very tight, indicating that the fitted GP has a high confidence regarding $\Delta(t,S)$. The bands get progressively wider at the edges.

To fully explain Figure~\ref{fig:bs}, we need to give the specification of the fitted GP described by \eqref{eq:gp_mean}. This includes the GP kernel $\kappa(\bx,\bx')$, the mean function $m(\bx)$, and the respective coefficients or hyperparameters $\vartheta$. In the figure, the mean function is $m(\bx) = \beta_0 + \beta_1 S =  -20.04 + 0.58 S$ and the kernel is squared-exponential~\eqref{eqn:se-kernel} with length-scales $\ell_1 = 0.626, \ell_2 = 10.00$, process variance $\sigma_p^2 = 239.71$ and noise variance $\sigma^2_\epsilon = 1.99 \cdot 10^{-4}$.  In this example, although the training outputs $Y^i$ are exact, for numerical purposes (namely to stabilize matrix inversion), we allow for a strictly positive observation noise $\sigma_\epsilon$ in \eqref{eq:noise}. The variance parameter was taken to be an unknown constant and learned as part of maximum likelihood estimation. %$\sigma_\epsilon$) was  $\sigma^2_\epsilon = 3.55707 \cdot 10^{-8}$.
The MLE was carried out using a genetic-based optimizer from package \texttt{rgenoud} in \texttt{R} and the overall GP fitting via the \texttt{DiceKriging}~\citep{DiceKriging} package.

\begin{remark}
{While the example above considers a very simple payoff, our approach trivially generalizes to arbitrary payoff structures, including portfolios of options with varying maturities and strikes. Since the surrogate construction is completely independent of the specifics of the price function $P(t,S)$ and samples $Y^i$, going from a Call option above to a collection of contracts with different $(T_j, K_j), j=1,\ldots, J$ only requires adjusting the code that provides the sample $Y^i$ while the rest proceeds as-is. Surrogates become attractive for computing sensitivities of large option portfolios even in analytic models, since they have a fixed evaluation cost, while the cost of evaluating a single $P(t,S)$ is linear in the number of contracts $J$ and becomes non-negligible for $J$ large.}
\end{remark}

\subsection{Observation Noise}
In the cases where $P(t,S)$ is not available exactly, the associated uncertainty will typically depend on $(t,S)$.
For real-life training datasets this would be due to varying bid-ask spreads that are driven by contract liquidity. For Monte Carlo based datasets, this would be due to the heteroskedastic conditional variance of the payoff as a function of $S$. For example, for a Call the conditional simulation variance $\sigma^2(\bx)$ tends to be higher in-the-money, since out-of-the-money nearly all empirical payoffs would be zero, so that $\sigma^2(\bx) \simeq 0$ for $S \ll K$.
GPs are able to straightforwardly handle non-constant $\sigma^2(\bx)$, this just requires replacing the term $\sigma^2_\epsilon \bm{I}$ with a diagonal matrix $\bm\Sigma$ with $\bm\Sigma_{ii} \equiv \sigma^2(\bx^i)$ in \eqref{eq:gp_mean}.
  For the case where the training set is model-based Monte Carlo (Setting b), we may estimate $\sigma^2(\bx^i)$ via the empirical standard deviation that corresponds to the empirical average for $Y^i$:
   \begin{align}\label{eq:inner-mc}
   Y^i &= \frac{1}{\check{N}} \sum_{j=1}^{\check{N}} \Phi (S^{j,i}_T), \\ \label{eq:emp-sigma}
   \hat{\sigma}^2(\bx^i) &= \frac{1}{\check{N}-1} \sum_{j=1}^{\check{N}} \left(\Phi(S^{j,i}_T) - Y^i \right)^2,
   \end{align}
   where $S^{j,i}_T \sim S_T | S_0 = S^i$ are i.i.d samples and $\Phi(S) = (S-K)_+$ is the Call payoff function.

   Plugging in $\diag \hat{\sigma}(\bx^i) \bm{I} $ into $\sigma^2_\epsilon \bm{I}$ in \eqref{eq:gp_mean}-\eqref{eq:gp_cov} is known as the Stochastic Kriging approach, see \cite{Ankenman}. This plug-in $\hat{\sigma}^2$ works well as long as $\check{N}$ is sufficiently large. The package \texttt{hetGP} \citep{BGL18} extends the idea of \eqref{eq:emp-sigma} to simultaneously learn $\sigma^2(\cdot)$ and $P(t,S)$ during the fitting step.

The baseline alternative is to assume a constant observation noise $\sigma_\epsilon$ which is augmented to the GP hyperparameters and estimated as part of MLE optimization. This is also the recommended approach for noiseless observations, where a small amount of noise (the so-called ``nugget'') is added (learned via MLE) in order to regularize the optimization of the other hyperparameters.

\subsection{Virtual Training Points}

The GP model has no a priori information about the properties of $P(t, S)$ and its fit is fully driven by the training data and the postulated prior mean $m(t,S)$. One way to improve the fit is by adding \emph{virtual} observations that reflect the structural properties. In particular, we can create ``boundary'' conditions by putting virtual (in the sense of not coming from any data) observations at the edges of the training space; see the right panel of Figure~\ref{fig:lv}. This ensures more stable and more confident estimates at extreme values of inputs.

Specifically, in our case studies below we:
\begin{itemize}
  \item Add virtual points deep-in-the-money to enforce $\widehat{\Delta}(t,S) \simeq 1$ in that region. This is achieved by adding $\tilde{y}^i = S^i- e^{-r(T-t^i)}K$
  at two close but distinct, large $S^i$'s.

  \item Add virtual points deep out-of-the-money to enforce $\widehat{P}(t,S) \simeq 0$ and therefore $\widehat{\Delta}(t, S) \simeq 0$. This is achieved by adding $\tilde{y}^i = 0$ for two close but distinct, small $S^i$'s.

  \item Add virtual points at contract maturity $\tilde{y}^i = (S^i-K)_+$ for $t^i = T$. This enforces the correct shape of $\widehat{P}(t,S)$ as $t \to T$, in particular the at-the-money kink of the Call payoff.
\end{itemize}

Above, we use virtual price observations; it is also possible to add virtual observations on the gradients of a GP, which however requires a much more involved model fitting.

\section{Case Studies}\label{sec:tests}

In this section we present the set up of two in-depth case studies, explaining the underlying stochastic models, implementation, and assessment metrics.

\subsection{Black--Scholes}

Our first test environment is a Black--Scholes model which provides a ground truth and hence ability to compute related exact errors. Moreover, we can generate arbitrary amount/shape of training data and observation noise.

We consider European Call options, priced via the classical Black--Scholes formula that also yields closed-form expressions for the Delta, Theta and Gamma. 
For the training data $Y^i$ we use Monte Carlo simulation of size $\check{N}$ and a plain sample average estimator:
\begin{align}\label{eq:bs-mc}
\mathbb{E}^Q[ e^{-r(T-t)}(S_T - K)_+] \simeq \frac{1}{\check{N}} \sum_{n=1}^{\check{N}} e^{-r(T-t)} (S^n_T - K)_+ =: Y^i,
\end{align}
where $S^n_T$ are i.i.d.~samples obtained using the log-normal distribution of $S_T$. While carrying out the Monte Carlo method (which is a proxy for any computationally heavy pricing engine), we also record the empirical standard deviation of the $\check{N}$ payoffs to obtain  the plug-in estimator $\hat{\sigma}(\bx^i)$ for the input-dependent noise variance parameter as in \eqref{eq:emp-sigma}.

We use $r=0.04$, $K=50$, $T=0.4$; for out-of-sample Delta hedging we assume asset $\mathbb{P}$-drift of $\mu = 0.06$ and initialize with $S_0 \sim {\cal N}(50,2)$.

\subsection{Local Volatility Model}
For our second case study we consider a nonlinear local volatility model, where Call price $P(t, S)$ is available only via Monte Carlo simulation. In this setup there is no direct ground truth and we obtain a pointwise ``gold standard'' estimate of $\Delta(t,S)$ through a large-scale, computationally expensive Monte Carlo simulation combined with a finite-difference approximation.

We consider the Local Volatility (LV) model  where the dynamics of $S$ under the physical measure is
\begin{align}\label{eq:sde}
dS_t = \mu S_t dt + \sigma(t,S_t) S_t dB_t,
\end{align}
where $B$ is a Brownian motion. Specifically, in numerical example in Section \ref{sec:lv}, we use the following piecewise local volatility function, see Figure~\ref{fig:lv}:
$$\sigma(t,S) := \begin{cases}
0.4 - 0.16 \ e^{-0.5 (T^*-t)}  \cos(1.25 \pi \log \frac{S}{S^*}), \qquad\mbox{ if }\; |\log \frac{S}{S^*}| < 0.4,\\
0.4,  \qquad\mbox{ if }\; |\log \frac{S}{S^*}| \geq 0.4,
\end{cases}$$
where $S^* = 50$ and $T^* = 0.4$.
The risk-free interest rate $r$ is set to 0.05 and the rate of return of $S$ is $\mu = 0.13$. In order to compute a gold-standard benchmark computation for the Delta, we use the central finite-difference approximation
$$
\widehat{\Delta}^{fd}(t,S; \delta) := \frac{ \widehat{P}(t,S+\delta) - \widehat{P}(t,S-\delta)}{2\delta},
$$
with discretization parameter $\delta = 0.01S_0$.
The two terms on the right hand side are computed via a Monte Carlo simulation with same stochastic shocks. Namely, we approximate $P(t,S\pm \delta)$ by the empirical average over $10^6$ paths simulated from the dynamics of $S$ described in (\ref{eq:sde}) with an Euler--Maruyama discretization with $\Delta t = T/100$ and with two different initial values $S_0 \pm \delta$, and the same sequence of randomly sampled $\Delta B_{t_i}$. One should notice that to implement this benchmark procedure in reality it is necessary to calibrate the local volatility function to the market data. This step is completely avoided with our GP methodology.

We again consider a Call option with strike $K=50$ and maturity up to $T=0.4$.

\begin{figure}[h!]
  			\centering \includegraphics[width=0.5\textwidth,trim=0.2in 0.2in 0.2in 1in]{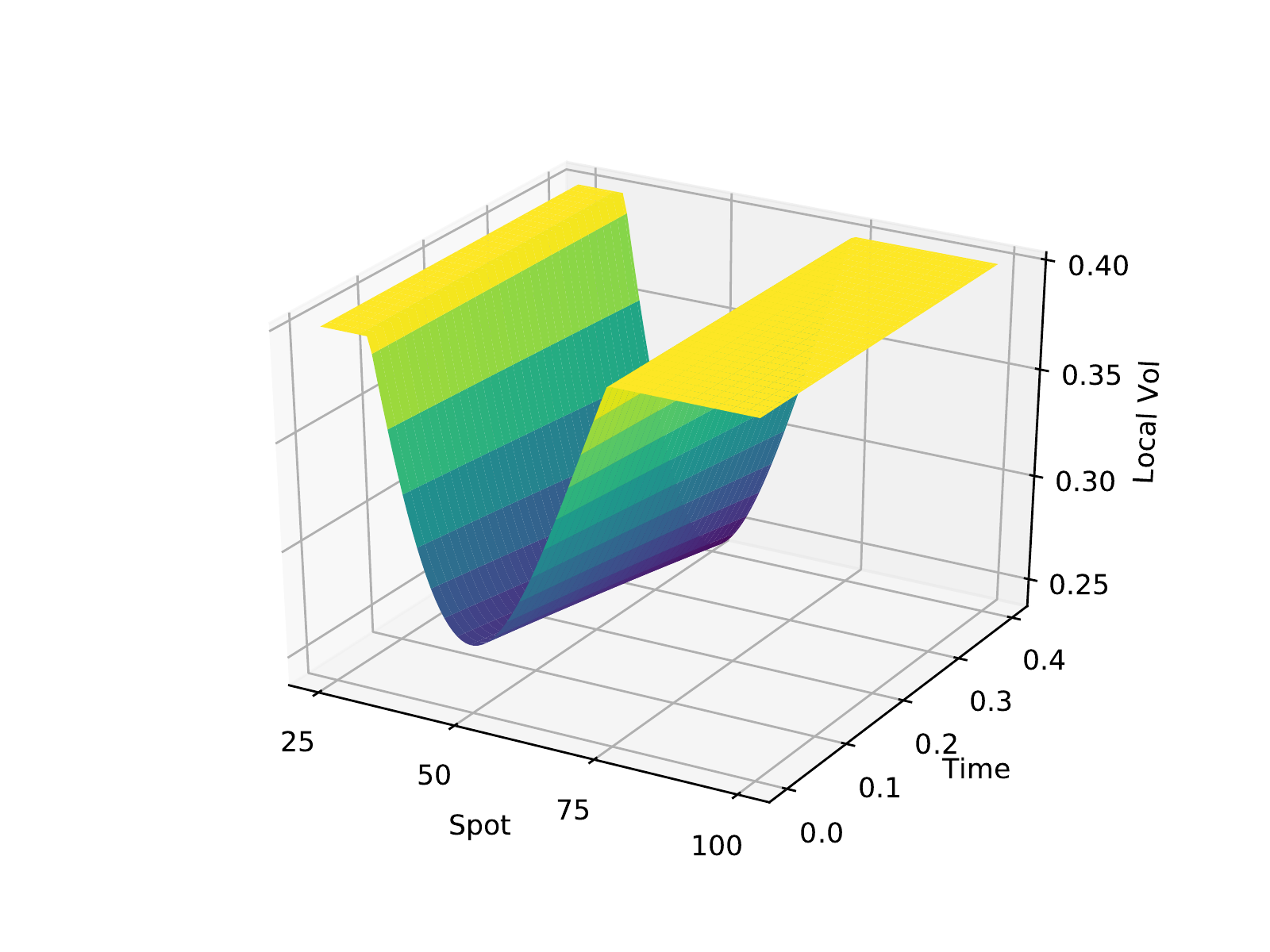} \includegraphics[width=0.45\textwidth]{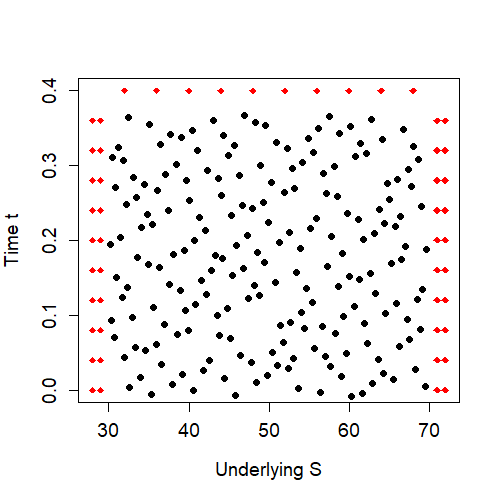}
\caption{ \label{fig:lv} Left: Local Volatility Surface $\sigma(t,S)$. Right: the training set (black $\cdot$'s) together with the virtual training points (red diamonds). Total design size is 250, with 200 points generated from a space-filling 2D Halton sequence on $[30,70] \times [-0.01, 0.37]$ and 50 virtual sites along 3 edges: 20 ITM, 20 OTM and 10 at maturity.}
\end{figure}

\subsection{Assessing Discrete Delta Hedging}\label{sec:delta-hedge}

To assess hedging quality, we implement a discrete-time delta hedging strategy which consists of rebalancing between the stock and the bank account based on a time step $\Delta t$ and the estimated $\widehat{\Delta}(t,S_t)$. We start the hedge at $t=0$ with the given wealth $W_0 = P(0,S_0)$ and update $W_t$ according to
\begin{align*}
  W_{t_k} = S_{t_k} \widehat\Delta(t_{k-1}, S_{t_{k-1}}) + (W_{t_{k-1}}-S_{t_{k-1}} \widehat\Delta(t_{k-1}, S_{t_{k-1}})) \cdot e^{r(t_k - t_{k-1})}.
\end{align*}
Repeating this along a discrete sequence of times $0 = t_0 < t_1 < \ldots < t_K = T$ we finally compare the payoff $\Phi(S_T)$ to the terminal wealth $W_T$, recording the resulting hedging error $E_T = W_T - \Phi(S_T)$. Note that due to time discretization, even though the market is complete in both case studies, we will have $E_T \neq 0$ almost surely and moreover the distribution of $E_T$ is affected by $\mathbb{P}$ since Delta hedging is done under the physical measure (stock drift is $\mu \neq r$).

Our primary comparator for hedging performance is the benchmark/true $\Delta(t, S)$. With continuous-time hedging, the latter yields an exact hedge (zero hedging error, almost surely) since the market is complete. 
The following proposition describes the hedging error $E_T$ when using $\widehat{\Delta}$ at discrete times.
\begin{proposition}\label{lemma:delta}
Under the local volatility model (\ref{eq:sde}), the hedging error when implementing the approximator $\widehat{\Delta}$ at discrete instants $0 = t_0 < t_1 < \cdots < t_K = T$ is given by
\begin{align}
E_T &= \underbrace{\sum_{k=0}^{K-1} \!\int_{t_k}^{t_{k+1}}(\Delta(t,S_t) - \Delta(t_k,S_{t_k})) dX_t}_{E_T^{(d)}} +\underbrace{\sum_{k=0}^{K-1} \! \int_{t_k}^{t_{k+1}}(\Delta(t_k,S_{t_k}) - \widehat{\Delta}(t_k,S_{t_k})) dX_t}_{\hat{E}_T},
\end{align}
where $dX_t = (\mu -r)S_tdt + \sigma(t,S_t)S_tdB_t$.
\end{proposition}

Proposition~\ref{lemma:delta} shows that the overall hedging error $E_T$ can be decomposed into two parts: the first, denoted $E_T^{(d)}$, is solely explained by the discrete-time aspect of the hedging strategy while the second one, $\hat{E}_T$ is driven by the approximation of the Delta, $\widehat{\Delta}$. Notice that $\hat{E}_T = \sum_{k=0}^{K-1} (\Delta(t_k,S_{t_k}) - \widehat{\Delta}(t_k,S_{t_k}))(X_{t_{k+1}} - X_{t_k})$. 
Since we are interested in studying the impact of $\widehat{\Delta}$ vs $\Delta$, we focus our analysis on $\hat{E}_T$. Taking the first and second moment, we have that $\mathbb{E}[E_T] = \mathbb{E}[E_T^{(d)}] + \mathbb{E}[\hat{E}_T]$ and $\mathrm{Var}(E_T) = \mathrm{Var}(E_T^{(d)}) + \mathrm{Var}(\hat{E}_T) + 2\mathrm{Cov}\left(E_T^{(d)}, \hat{E}_T \right)$. The next Corollary addresses the contribution from $\hat{E}_T$.

\begin{corollary}\label{cor:moments}
The mean and variance of $\hat{E}_T$ are given by
\begin{align}\label{eq:ET}
&\mathbb{E}[\hat{E}_T] =  (\mu - r) \sum_{k=0}^{K-1} \int_{t_k}^{t_{k+1}}\mathbb{E}\left[(\Delta(t_k,S_{t_k}) - \widehat{\Delta}(t_k,S_{t_k})) S_t\right] dt ,\\ \label{eq:VarE}
&\mathrm{Var}(\hat{E}_T) = \sum_{k=0}^{K-1} \int_{t_k}^{t_{k+1}}  \mathbb{E}\left[(\Delta(t_k,S_{t_k}) - \widehat{\Delta}(t_k,S_{t_k}))^2 \sigma^2(t,S_t) S_t^2 \right]dt.
\end{align}
Moreover, the covariance between $E_T^{(d)}$ and $\hat{E}_T$ is
\begin{align*}
&\mathrm{Cov}(E_T^{(d)}, \hat{E}_T) = \sum_{k=0}^{K-1} \int_{t_k}^{t_{k+1}}  \mathbb{E}\left[(\Delta(t,S_t) - \Delta(t_k,S_{t_k}))(\Delta(t_k,S_{t_k}) - \widehat{\Delta}(t_k,S_{t_k})) \sigma^2(t,S_t) S_t^2 \right]dt.
\end{align*}
\end{corollary}

The proofs of Proposition~\ref{lemma:delta} and Corollary~\ref{cor:moments} can be found in Appendix~\ref{app:approximated_delta}.

\begin{remark}\label{rem:varE}

Suppose $\widehat{\Delta}$ is an unbiased estimator of $\Delta$ in the following sense:
$$\mathbb{E}[\widehat{\Delta}(t_k,S_{t_k})|S_{t_k}] = \Delta(t_k,S_{t_k}), \quad\mathbb{E}[(\Delta(t_k,S_{t_k}) - \widehat{\Delta}(t_k,S_{t_k}))^2|S_{t_k}] = \sigma^2,$$
and, for any $t \in (t_k,t_{k+1}]$ and any bounded Borel-measurable functions $\phi$ and $\psi$,
$$\mathbb{E}[\psi(\Delta(t_k,S_{t_k}) - \widehat{\Delta}(t_k,S_{t_k})) \phi(S_t)|S_{t_k}] = \mathbb{E}[\psi(\Delta(t_k,S_{t_k}) - \widehat{\Delta}(t_k,S_{t_k}))|S_{t_k}] \mathbb{E}[\phi(S_t)|S_{t_k}],$$
for every $k \in \{0,\ldots,K-1\}$. In that case, conditioning on $S_{t_k}$ we find
\begin{align*}
\mathbb{E}[\hat{E}_T] = 0, \,
\mbox{Var}(\hat{E}_T) = \sigma^2 \mathbb{E}\left[\langle S \rangle_T \right]  \mbox{ and } \mathrm{Cov}(E_T^{(d)}, \hat{E}_T)  = 0.
\end{align*}
Thus, for unbiased $\widehat{\Delta}$, we expect to see no additional hedging loss and additional hedging variance that is proportional to the approximation variance. In other words, good $\Delta$ approximators should not impact expected hedging loss; while the mean-squared error of $\widehat{\Delta}$ is a proxy for the variance of the hedging loss.
\end{remark}

\begin{remark}\label{rem:continuous}

Under continuous-time Delta hedging, we have $E_T^{(d)} = 0$ and $E_T = \int_0^T (\Delta(t,S_t) - \widehat{\Delta}(t,S_t)) dX_t$. Moreover, if the true Delta is known but hedging is done discretely in time, then $\hat{E}_T = 0$ and $E_T = E_T^{(d)}$.
\end{remark}

Another delta-hedging strategy is the so-called implied Delta hedging which relies on the  Black--Scholes delta with the current implied volatility. Let $IV(t,S)$ denote the implied volatility satisfying
$$P(t,S) =: {P}^{BS}(t,S,IV(t,S)),$$
where ${P}^{BS}(t,S, \sigma)$ is the Black--Scholes formula price for this option with volatility $\sigma$. Then Implied Delta is
\begin{align}
  \Delta_I(t,S) := \Delta^{BS}(t,S, IV(t,S)).
\end{align}
Note that the following is true
$$\Delta(t,S) = \Delta_I(t,S) + \mathcal{V}^{BS}(t,S, IV(t,S)) \frac{\partial IV}{\partial S}(t,S),$$
where $\mathcal{V}^{BS}$ is the Black--Scholes Vega. So the difference between the true and implied delta is linked to the option Vega and the implied volatility skew. Practically speaking, the implied volatility is a local average of $\sigma(t,S)$. For the local volatility case study, the implied Delta is too low OTM and too high ITM, generating a non-negligible hedging error as a result. The latter feature demonstrates the importance of properly learning option sensitivities, rather than just calibrating the immediate implied volatility surface.

In terms of Proposition \ref{lemma:delta} we note that the Implied Delta is not unbiased and the variance of the approximation part of the hedging error can be written as
\small
\begin{align}
\mathrm{Var}(\hat{E}^I_T) &= \sum_{k=0}^{K-1} \int_{t_k}^{t_{k+1}} \mathbb{E}\left[ \mathcal{V}^{BS}(t_k,S_{t_k}, IV(t_k,S_{t_k}))^2 \left(\frac{\partial IV}{\partial S}(t_k,S_{t_k})\right)^2 \sigma^2(t,S_t) S_t^2  \right]dt.
\end{align}
\normalsize

\subsection{Performance Metrics}\label{sec:assess}
Let us denote by $\widehat{\Delta}(t,S)$ a given GP-based estimate of $\Delta(t,S)$. We define the following performance metrics to assess the quality of $\widehat{\Delta}(t,S)$. In all cases, $\cD'$ refers to a discrete test set {of size $N'$}.

\textbf{Metric I: RIMSE.} Assuming that a gold-standard (possibly exact) $\Delta(t,S)$ is available, we compare $\widehat{\Delta}(t,S)$ to the ground truth $\Delta(t,S)$. Our main choice is the root integrated mean-squared error (RIMSE) defined as:
\begin{align}\label{eq:rimse}
  \mathrm{RIMSE}^2 := \frac{1}{N'}\sum_{(t,S) \in \cD'} (\widehat{\Delta}(t,S) - \Delta(t,S))^2,
\end{align}
for a test set $\cD'$. RIMSE is the standard $L_2$ criterion for judging the quality of $\widehat{\Delta}$ over a region of interest. We can similarly define the RIMSE for greek/sensitivity  $\Theta$ and for the option price $P$ itself, denoted as $\widehat\Theta_{Err}$ and $\widehat{P}_{Err}$.

\textbf{Metric II: PnL.} We also measure the quality of $\widehat{\Delta}(t,S)$ directly through the Delta-hedging P\&L. {We report the variance of the terminal
P\&L $E^n_T = \Phi(S^n_T)-W^n_T$. $Var(E_T)$ is available even without a ground truth.
Better Delta forecasts should lead to lower variability of hedging errors, but $Var(E_T)$ is always bounded away from zero due to time discretization.}

\textbf{Metric III: MAD.} We observe that the Greek estimators tend to have a few small regions of large errors around the edges of $\cD'$ which inflates RIMSE in \eqref{eq:rimse}. To mitigate this effect, we evaluate the Median Absolute Deviation (MAD) metric
\begin{align}\label{eq:mad}
  \mathrm{MAD} := \mathrm{Median}_{(t,S) \in \cD'} |\widehat{\Delta}(t,S) - \Delta(t,S))|,
\end{align}
where the median is over the discrete test set $\cD'$. Thus, the $L_1$ approximation error will be less than MAD at half the test sites.

\textbf{Metric IV: Coverage.} To assess the uncertainty quantification provided by the GP model, we evaluate the accuracy of the associated credible bands. Specifically, a good model will be close to matching the nominal coverage of its bands, i.e.~the ground truth should be within the 95\% credible bands at 95\% of the test locations, cf~\eqref{eq:ci}:
\begin{align}\label{eq:coverage}
\mathrm{Cvr} := \frac{1}{N'} \sum_{(t,S) \in \cD'} 1_{ \left( {\Delta}(t, S) \in \mathrm{CI}_{0.95}(t, S) \right)}.
\end{align}
A model with $\mathrm{Cvr} < 0.95$ has overly narrow credible bands and a model with $\mathrm{Cvr} > 0.95$ has them too wide.

\textbf{Metric V: NLPD.} The Negative Log Probability Density metric blends the testing of the posterior mean (via MSE) and of the posterior standard deviation:
\begin{align}\label{eq:nlpd}
  \mathrm{NLPD}(t,S) := \frac{ (\Delta(t,S) -\widehat{\Delta}(t,S))^2}{V_g(t,S)} + \log V_g(t,S).
\end{align}
where $V_g$ is the posterior variance of the Delta estimator, see Section \ref{sec:greek}. Better models will have lower NLPD. NLPD can be viewed as combining RIMSE and Coverage.

\textbf{Metric VI: Bias.} To assess whether the estimator tends to consistently over- or under-estimate the true Greek, we record its statistical bias:

\begin{align}
  \mathrm{Bias} := \frac{1}{N'} \sum_{(t,S) \in \cD'} (\widehat{\Delta}(t,S) - \Delta(t,S)).
\end{align}
Since our $\widehat{\Delta}$ are statistically constructed, we expect minimal bias.

\textbf{Metric VII: Empirical Moments of $\hat{E}_T$. } Reflecting Corollary~\ref{cor:moments} we evaluate the following two quantities related to the hedging loss:
\begin{align}\label{eq:muE}
  \mu_E &:= (\mu-r) \cdot T \cdot \frac{1}{N'} \sum_{(t,S) \in \cD''} [\widehat{\Delta}(t,S) - \Delta(t,S)]S f_{S_t}(S); \\ \label{eq:VE}
  V_E &:= T \cdot \frac{1}{N'} \sum_{(t,S) \in \cD''} (\widehat{\Delta}(t,S) - \Delta(t,S))^2 \sigma(t,S)^2 S^2 f_{S_t}(S),
\end{align}
where $f_{S_t}(\cdot)$ is the probability density function of $S_t$. Thus, $\mu_E$ is an empirical proxy for the average extra hedging loss $\mathbb{E}[\hat{E}_T]$ due to $\widehat{\Delta}$ and $V_E$ is an empirical proxy for the respective additional hedging variance $\mathrm{Var}[\hat{E}_T]$. Good models should have $\mu_E \simeq 0$ and low $V_E$.

We generate $\cD''$ by forwarding simulating $(S_t)$ trajectories which allows us to drop the $f_{S_t}$ term. In both case we use $\Delta t = 0.02$ to sum over $t_k$ in
%In both cases, the integral over $t$ in
\eqref{eq:ET}-\eqref{eq:VarE} .

\medskip
For the Black--Scholes case study we use a test set of $N'=|\cD'| =1600$ sites constructed as a grid on $\{-0.01, 0.01, \ldots, 0.37 \} \times \{30, 30.5, \ldots, 69.5\}$ and a nominal coverage level of 95\% for \eqref{eq:coverage}. For the LV case study we use a test set of $N'=341$ sites with 11 time-steps $t \in \{0, 0.04, \ldots, 0.36, 0.4\}$ and 31 stock price levels $S \in \{29.4, 31.03, \ldots, 78.4\}$. {We do not report running times since those are highly dependent on the hardware used, as well as $N$, number of inner Monte Carlo simulations $\check{N}$, number of test locations $N'$, and the complexity of the contract (or portfolio of contracts). As a guide, using a \texttt{R}-based prototype implementation on a 2018-vintage laptop, it takes a dozen of seconds to fit a GP model with $N=200$. It then takes another handful of seconds to evaluate the Greeks on the above test set. For larger training sets with $N \simeq 400$, fitting takes a bit over a minute. Significant hardware and software improvements are feasible for an industrial-grade deployment.}

\section{Results}\label{sec:bs}

In this section we present the experimental results based on the two case studies described above. We start with the Black--Scholes set up where ground truth is known and training inputs are noisy due to a Monte Carlo approximation.

\subsection{Choice of GP Kernel}
We first consider the impact of different GP model components on the quality of the Delta approximation. We begin with the role of the kernel family which is the most important choice to be made by the user. To do so, we compare the use of SE, M52 and M32 families, each of which is fitted in turn via MLE. Recall that these three families imply different degree of smoothness in $\hat{P}$ (and hence in the fitted Greeks): squared-exponential kernel will lead to very smooth fits, while Mat\'ern kernels allow more roughness.

Figure~\ref{fig:kernel-bs} shows the fits and 95\% credible bands across the above 3 kernel families and different Greeks.  The results are further summarized in Table~\ref{tbl:kernel}. While the training is done jointly in the $t$ and $S$ dimension, we illustrate with one-dim plots that fix $t$ and show dependence in $S$ only.

The top left panel shows the error $\widehat{P}(t,S)-P(t,S)$ between the fitted and true option prices and therefore provides an immediate sense of the accuracy of the statistical surrogate. We observe that all three GP models perform well out-of-the-money (OTM) and the largest error is in-the-money (ITM). This phenomenon is driven by the higher conditional variance of training inputs $Y^i$ ITM where Monte Carlo estimates are less accurate. In essence, the observation noise is proportional to the price and hence estimating the latter is harder when $P(t,S)$ is higher.

\begin{figure}[!ht]
	\begin{tabular}{cc}
	\begin{minipage}[t]{0.45\linewidth}
		\begin{center}
			\includegraphics[width=1\textwidth,height=2.5in]{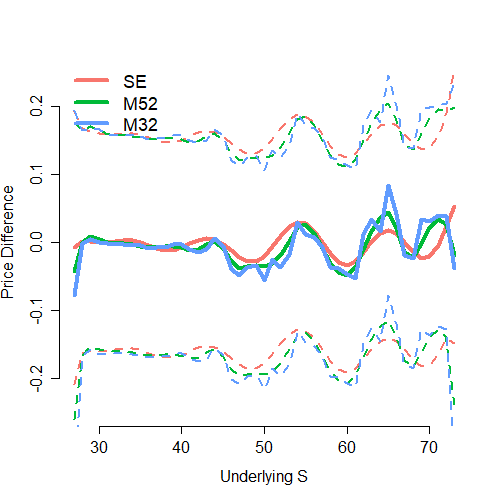}
		\end{center}
	\end{minipage} &
	\begin{minipage}[t]{0.45\linewidth}
		\begin{center}
			\includegraphics[width=1\textwidth,height=2.5in]{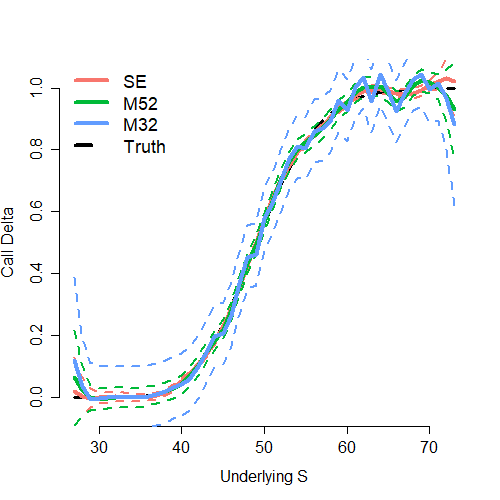}
		\end{center}
	\end{minipage} \\
Call Price Error & Delta $\partial P/\partial S$\\
\begin{minipage}[t]{0.45\linewidth}
		\begin{center}
			\includegraphics[width=1\textwidth,height=2.5in]{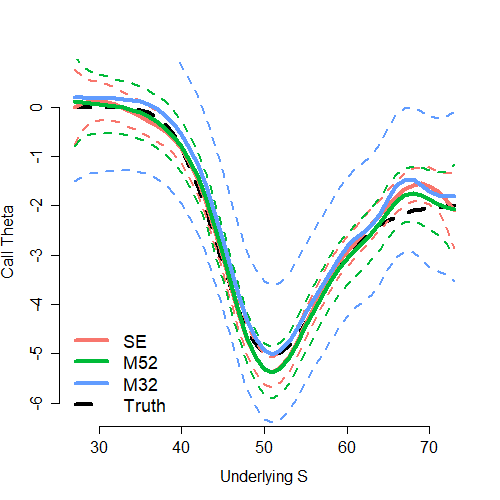}
		\end{center}
	\end{minipage} &
	\begin{minipage}[t]{0.45\linewidth}
		\begin{center}
			\includegraphics[width=1\textwidth,height=2.5in]{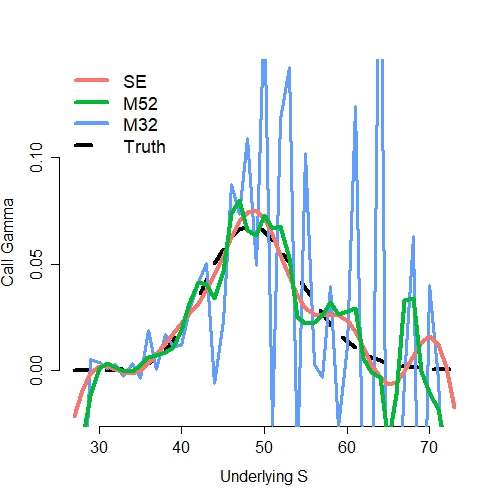}
		\end{center}
	\end{minipage} \\
Theta $-\partial P/\partial t$ & Gamma $\partial^2 P/\partial S^2$
 \end{tabular}
	\caption{Estimated sensitivities at $t=0.1$ together with their 95\% credible bands for a Black--Scholes Call across three different GP kernel families and using a space-filling experimental design. The Gamma is computed using a finite difference approximation. Ground truth indicated in dashed black line. Training set of $N=400$ inputs and $\check{N}=2500$ inner MC simulations.
	\label{fig:kernel-bs}}
\end{figure}

The top right panel displays the resulting $\widehat\Delta(t, \cdot)$'s which are simply the gradients of the respective surrogates $\widehat P(\cdot, \cdot)$ with respect to the second coordinate. In general, all three kernels perform very well, closely matching the true Delta. %We observe that the SE kernel overshoots $\widehat{\Delta}(t,S)>1$ around $S=50$, which is caused by its tendency to over-smooth and therefore it is not able to rapidly lower the convexity during the transition from ATM to ITM. In contrast, the Mat\'ern kernels can ``follow'' the Delta more flexibly. This feature also manifests itself in
We observe the very narrow credible bands of the SE kernel compared to the M52 and M32 ones, with the latter having the widest credible band. The general observation (well known in the surrogate literature) is that the smoother is $m_*(\cdot)$, the tighter the CI. Consequently, all the CIs of a SE-family GP will always be narrower compared to M52 or M32. %We also note the ``pinch points'' around $S=71$ and $S=29$ caused by the virtual training points at those locations.
The other feature we see is oscillations of the M32-based Delta deep-ITM, and moreover that all models exhibit reversion to the prior beyond the edge of the training set, manifested by $\widehat{\Delta}(t,S) < 1$ for $S \gg 70$ and $\widehat{\Delta}(t,S) > 0$ for $S\ll 30$. The virtual training points are critical in avoiding this issue and enforce $\widehat{\Delta}(t,S) \simeq 1$ around $S=70$ and $\widehat{\Delta}(t,S) \simeq 0$ for $S\simeq 30$.

\begin{table}[htb]
	\caption{Effect of the GP kernel family on learning the Delta in a Black--Scholes model. We report 7  metrics for $\widehat{\Delta}$, as well as the RIMSE for $\Theta$ and option price $P$ (last 2 columns, cf.~\eqref{eq:rimse}). All metrics  are based on a gridded test set of $|\cD'| = 80 \times 20 = 1600$ sites.
	\label{tbl:kernel}}
	\begin{center} \vspace*{-0.2in}
				{ $$\begin{array}{lrrrrrrrrr}
					\hline
\text{Kernel} & \text{RIMSE} &  \text{MAD} & 95\% \text{Cvr} &  \text{Bias} & \text{NLPD} & \mu_E & V_E & \widehat\Theta_{Err} & \widehat{P}_{Err}  \\
					\hline
\text{SE}  & 0.0134 & 0.0070 & 0.7281 & 0.00043 & -6.32 & -0.00039 & 0.0504 & 0.597 & 0.028\\
\text{M52} & 0.0165 & 0.0075 & 0.9550 & -0.00004 & -7.25 & 0.00031 & 0.0484 & 0.674 & 0.031\\
\text{M32} & 0.0298 & 0.0124 & 0.9888 & 0.00009 & -5.58 & -0.00024 & 0.0771 & 0.753 & 0.035\\ \hline
				\end{array}$$}
	\end{center}\vspace*{-0.2in}
\end{table}

The bottom left panel illustrates the fitted $\widehat\Theta(t,\cdot)$ which uses  the exact same GP models as in the first row of the figure, simply computing the gradient in the other coordinate. This is one of the advantages of our framework---once fitted, all sensitivities across the different coordinates are obtained in the same consistent manner. Due to the more complex shape of the Theta, and in particular higher convexity of $P(t,S)$ in $t$, the quality of $\widehat{\Theta}$ is poorer compared to that of Delta. Both the SE and M52 overestimate the steep peak of $\Theta$ ATM, estimating $\widehat\Theta(0.1,50) \simeq -5.5$ rather than the true $-5$. We note that the M52/M32 surrogates are aware of this challenge and provide appropriately wide CI bands that contain the ground truth (in fact the M32 band is too wide). In contrast, the SE surrogate overestimates its posterior uncertainty, with the result that its coverage for $\Theta$ is much below the nominal 95\% level (i.e.~the CI frequently does not contain the ground truth). Another region where all models exhibit lack of fit is for $S \in [60,70]$.

Finally, the bottom right panel of Figure~\ref{fig:kernel-bs} illustrates the estimation of $\Gamma(t,S)$. Numerical estimation of second-order sensitivities is extremely challenging, especially through functional approximators. In that light, the SE and M52 GP surrogates perform quite well given that they were trained on just 400 noisy observations. We do observe significant oscillations in $\widehat\Gamma$ especially for $S \gg 55$, which is not surprising since the original $\widehat P$'s are not constrained in any way and tend to wiggle or vibrate in the input space. The oscillations are mild for the SE kernel (again, due to the tendency to over-smooth spatially) and are very severe for M32. We note that mathematically $m_*(\cdot)$ is only ${\cal C}^1$ for the Mat\'ern-3/2 family, and so there is actually \emph{no} second-order sensitivity for this surrogate. In the plot we obtain an approximation through finite differences, cf.~\eqref{eq:fd}, which are in fact the reason for the sharp oscillations. 
Table \ref{tbl:kernel} reports the error metrics defined in Section~\ref{sec:assess} for the above three surrogates.  We concentrate on the estimation of $\Delta$ (where we report 7 different metrics), as well as report the RIMSE for $\Theta$ and for option price itself, $P$. The surrogate utilizing a SE kernel appears to be best in terms of integrated mean squared error and also has slightly lower median absolute deviation and lower bias. However, it also has poor coverage suggesting that it is overconfident and reports too narrow credible bands. This is confirmed by the NLPD score that is worse than that for M52-based surrogate. The latter also beats SE in terms of RIMSE for $\Theta$ and essentially yields the same RIMSE for the price $P$. The M32-based surrogate is worst across the board, and also overestimates uncertainty (its coverage is much higher than 95\%).

To summarize, there are two key take-aways. On the one hand, the SE kernel \eqref{eqn:se-kernel} tends to over-smooth and therefore has trouble reproducing the spatial non-stationarity one observes for most option payoffs (namely high convexity ATM and almost linear deep ITM and deep-OTM). It also underestimates posterior uncertainty. On the other hand, the Mat\'ern-3/2 kernel tends to give CIs that are \emph{too} wide and by its nature is a very poor choice for second order sensitivities, like Gamma. In light of above, we recommend to use the Mat\'ern-5/2 kernel which provides the best compromise in terms of maximizing RIMSE and MAD, minimizing NLPD, and matching coverage.

\subsection{Size of Training Set}

Next, Table~\ref{tbl:n} shows how the size $N$ of the experimental design affects the fit. Naturally, a larger training set $\mathcal{D}$ provides more information and hence should yield a better fit. Consequently, larger $N$ should imply lower error metrics across the board (apart from the Coverage statistic that should converge to its nominal $95\%$ level).

\begin{table}[ht]
	\caption{Effect of training set size on learning the Delta and other Greeks in a Black--Scholes model. We report 8  metrics for $\widehat{\Delta}$, as well as the RIMSE for $\Theta$ and option price $P$ (last 2 columns). All metrics  are based on a gridded test set of $80 \times 20 = 1600$ sites, $\{ S_0: 30, 30.5, \ldots, 69.5\} \times \{t : -0.01, 0.01, \ldots, 0.37\}$. Training is based on \eqref{eq:bs-mc} with $\check{N}=2500$ inner simulations, plus 50 virtual training points, and the GP surrogates have Mat\'ern-5/2 kernel, linear trend function and estimated constant $\sigma_\epsilon$. {The reference hedging variance $Var(E_T)$ (7th column) using exact Delta is 0.265.}\vspace*{-0.2in}
	\label{tbl:n}}
	\begin{center}%\vspace*{-0.1in}
		{\small $$\begin{array}{rrrrrrrrrrrrr}
					\hline \hline
					N & \text{RIMSE} &  \text{MAD} &   95\% \text{Cvr} &  \text{Bias} & \text{NLPD} & {Var(E_T)} & \mu_E & V_E &   \widehat\Theta_{Err} & \widehat{P}_{Err}  \\ \hline
80 &  0.0318 & 0.0136 &  0.9212 & -0.00044 & -5.884&  0.321 & 0.00079 & 0.0759 &0.753 & 0.054\\
120 &0.0183 & 0.0092 &  0.9900 & -0.00039 & -6.911 &  0.307 & 0.00004 & 0.0541&  0.653 & 0.034\\
160 &0.0157 & 0.0077 &  0.9869 & -0.00078 & -7.193 &  0.300 & -0.00022 & 0.0490& 0.670 & 0.032\\
200 &0.0150 & 0.0075 &  0.9856 & 0.00030 & -7.278 &   0.301 & -0.00015 & 0.0491&0.659 & 0.032\\
240 &0.0192 & 0.0099 &  0.9338 & 0.00057 & -6.919 &   0.311 & -0.00030 & 0.0616&0.677 & 0.034\\
280 &0.0175 & 0.0086 &  0.9650 & 0.00053 & -7.113 &   0.306& -0.00021 & 0.0570&0.671 & 0.033\\
320 &0.0154 & 0.0070 &  0.9812 & 0.00091 & -7.352 &   0.299 & -0.00027 & 0.0520&0.676 & 0.031\\
360 &0.0150 & 0.0067 &  0.9794 & 0.00052 & -7.410 &   0.300 & -0.00033 & 0.0528&0.661 & 0.030\\
400 &0.0142 & 0.0059 &  0.9800 & 0.00030 & -7.504 &   0.299& -0.00017 & 0.0517&0.663 & 0.028\\
\hline
\end{array}$$}
\end{center}\vspace*{-0.2in}
\end{table}

This pattern is generally observed in Table~\ref{tbl:n}; we find a roughly ${\cal O}(N^{-1/2})$ rate for RIMSE and MAD (both for Delta, as well as for Theta and Price, see the last two columns). The above trend is quite noisy because learning is not necessarily monotone in $N$ since the estimated GP hyperparameters change across datasets. As a result it is possible that a surrogate with higher $N$ has worse performance, compare $N=200$ and $N=240$ in Table~\ref{tbl:n}. This occurs because the estimation errors in GP surrogates tend to arise via small spurious oscillations in the predicted response in regions with sparse training data. As $\cD$ expands, those oscillations can shift abruptly as the MLE optimizer finds new local maxima for the hyperparameters.

 One very reassuring finding is that all surrogates are unbiased in their estimates of $\Delta$, even for very low $N$. Another feature we observe is that learning $\Theta$ is more challenging, with the respective RIMSE converging quite slowly. This is linked to the spatial nonstationarity, namely the fact that $S \mapsto \Theta(t,S)$ changes rapidly ATM but slowly ITM/OTM, and moreover goes to $-\infty$ at-the-money at maturity.

Another important observation is that the patterns in all the considered metrics (beyond NLPD/Coverage) are broadly similar and therefore RIMSE is a good overall proxy for approximation quality. In that sense, the standard mean squared error is sufficient for assessment of the point predictions for the Greeks; NLPD is a good complement for assessing uncertainty quantification.

\subsection{Simulation Design}

The GP surrogate is a data-driven spatial model and consequently is sensitive to the geometry of the training set. Therefore, we analyze the impact of the \emph{shape} of $\cD$, whose choice is entirely up to the modeler, on the quality of the Greeks approximation. 

The spatial covariance structure driven by $\kappa(\cdot, \cdot)$ implies that for a given $(t,S)$,  $\widehat{\Delta}(t, S)$ is primarily determined by the training points in its vicinity. Consequently, to ensure a good \emph{average} approximation quality, it is desirable to \emph{spread} the training points, namely $\cD$ should reflect the test set $\cD'$. The respective concept of a \emph{space-filling} experimental design can be achieved in multiple ways. One obvious candidate is a gridded design, putting $\{t^i, S^i\}$ on a two-dimensional lattice. A gridded $\cD$ can however interfere with fitting of a Gaussian process model, because only a few values of distances $|x_j-x_j'|$ used within $\kappa(\bx,\bx')$ are then observed, making learning of the lengthscales more difficult. On the flip side, a gridded $\cD$ makes $\bK$ of (\ref{eq:gp_mean}) a Kronecker matrix, which can be exploited for computational speed-ups \citep{flaxman2015fast,wilson2015kernel}.

As an alternative to a training grid, one can utilize space-filling sequences, either deterministic low-discrepancy sequences, such as the (scrambled) Sobol and Halton sequences used widely in the Quasi Monte Carlo (QMC) literature \citep{Lemieux2009} or Latin Hypercube Sampling (LHS). LHS yields randomized designs that is effectively variance-reduced i.i.d.~Uniform sampling. Both approaches allow to specify a training set $\cD$ of arbitrary size. We find that the choice of \emph{how} to space-fill plays limited role in overall performance and generically employ Halton sequences in subsequent experiments. Space-filling also generalizes to higher dimensions where gridding becomes infeasible.

A related aspect concerns the impact of simulation noise on learning the Greeks. A natural question is whether it is better to train on a few highly-accurate data points, or on many low-precision inputs. This corresponds to the trade-off between design size $N = |\cD|$ and the number of MC samples $\check{N}$ in \eqref{eq:inner-mc} (see also \eqref{eq:bs-mc}). Figure~\ref{fig:paths} visualizes RIMSE of $\widehat{\Delta}$ as we vary $N,\check{N}$. We observe limited gains from increasing $\check{N}$, so the spatial effect dominates and the quality of the Delta approximation depends primarily on having a large (in terms of many different $S$-values) training set. We also note the large improvement in fit quality when the GP model switches from smoothing + interpolation to pure interpolation (the case where training inputs are exact). Indeed we see that using $N=100$ exact training points is better than training with $N=500$ inputs observed in slight noise $\check{N}=16,000$.

\begin{figure}[!ht]
\centering
\includegraphics[width=0.45\textwidth,trim=0.2in 0.2in 0in 0.2in]{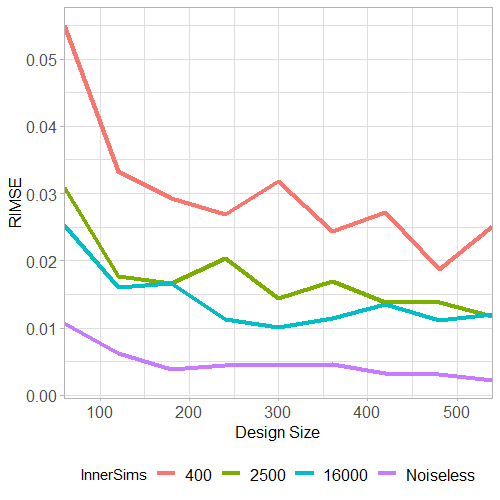}
   \includegraphics[width=0.53\textwidth,trim=0.1in 0.2in 0.1in 0.2in]{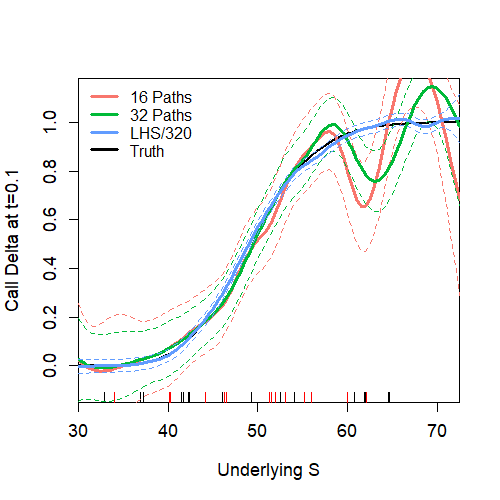}
  \caption{\emph{Left:} Impact of simulation design: root integrated mean squared error as a function of design size $N$ and number of inner MC simulations $\check{N}$.  \emph{Right:}  Comparing Delta approximation based on a space-filling design with 320 inputs to one based on 16 paths ($\Delta t =0.02$, 320 training inputs) and 32 paths (640 training inputs). All designs are for the Black--Scholes Call case study and are augmented with additional 50 virtual training points.
  \label{fig:paths}}
\end{figure}

\subsection{Quality of Delta Hedging}

Recall from Section~\ref{sec:delta-hedge} that we may decompose the total hedging loss $E_T$ into a component $E^{(d)}_T$ due to time discretization,  and a component $\hat{E}_T$ due to the Delta approximation error.  Taking the representative M52 model from Figure~\ref{fig:kernel-bs} for the Black--Scholes case study, and using $n=1,\ldots, 2500$  scenarios $(S^n_t)_{t \in [0, T]}$,  with 20 hedging periods $\Delta t= 0.02$ and $S_0 \sim {\cal N}(50,2^2)$, we find that the resulting hedging error has $\mathrm{Ave}(E_T) = 0.0163$ and $\mathrm{Var}(E_T) = 0.2980$. In comparison, hedging with the exact Black--Scholes $\Delta$ on the same set of paths we estimate $\mathbb{E}[ E^{(d)}_T ] =0.0145 $ and $\mathrm{Var}(E^{(d)}_T) = 0.2650$. Thus, in both cases hedging errors are effectively mean-zero and there is no additional bias from $\widehat{\Delta}$. Moreover,
as expected errors in $\widehat{\Delta}$ increase the variance of the hedging error; in this example they add about 3 cents of standard deviation ($\mathrm{StDev}(E^{(d)}_T) = 0.5148, \mathrm{StDev}(E_T) = 0.5459$) or about 6\% of the original. Finally, we obtain $V_E = 0.0483$ and $\mu_E = -3.1 \cdot 10^{-4}$ which is quite consistent with Corollary 1 and \eqref{eq:muE}-\eqref{eq:VE}, namely that $\mathrm{Var}(E_T) \simeq \mathrm{Var}( E^{(d)}_T ) + V_E$. {As hedging quality increases, we observe the strongest effect on the tail of $E_T$. For example, in Table~\ref{tbl:n} we report the one-sided $L_1$ hedging loss for the Call option. We observe strong improvements as training set gets larger and surrogate quality improves. On the other hand, very limited gains would be recorded if we report the $L_1$ or $L_2$ norm of $E_T$.}

\section{Path-Based Training}\label{sec:discuss}

A further motivation for the task of estimating the Greeks based on a sparse set of price data is the case where the training set $\cD$ is the \emph{history} of the contract price $Y_i = P(t_i, S_{t_i})$ along trajectories of the underlying $S_{t_0}, S_{t_1}, \ldots, S_{t_n}$. The latter is interpreted as historical observations, i.e.~a model-free paradigm where one directly uses data to learn price sensitivities. In this setting the training set $\cD$ is fixed and depends on how much data the modeler was able to collect. Clearly, a single trajectory would be insufficient for good inference; one typically would consider expired options with same strike, indexing data by time-to-maturity $\tau = T-t$ of the contract. (Under additional assumptions, one may also switch from asset price $S$ to log-moneyness $S/K$ that allows to simultaneously consider options with multiple strikes.) The resulting training sample is limited by the fact that asset time series tend to be non-stationary over long periods. This setting naturally suggests
the possibility of dynamically updating $\cD$ as more historical data is collected, see Section \ref{sec:updating} below.

Path-based training makes $\cD$ to have an irregular pattern in the $S$ dimension.
In the right panel of Figure~\ref{fig:paths} we
investigate the resulting impact on Greek approximation quality, by training our GP surrogate on  a collection of $(S_t)$-paths, sampled at some fixed time frequency $\Delta t$. The plot  shows Deltas fitted on two different datasets: one generated on a grid of $(t,S)$ values as in the previous section, and another sampled at a regular sequence of $t$'s, but along paths of $(S_t)$. In the latter case $\cD = \{ (t^j_{i}, S_{i\Delta t}^j) : t^j_i = i\Delta t \}$ for $j=1,\ldots J$ with $(S^j_\cdot)$ being $J$ i.i.d.~paths of $S$ started at pre-specified initial locations $S_0^j$.

We observe that training using paths is significantly inferior relative to training using a space-filled design. The path-based $\cD$ tends to have a lot of ``holes'' where the model is unable to accurately ``see'' the gradient. This leads to worse estimates of the GP hyperparameters $\vb$, as well as in wider credible bands.  We find that without a lot of fine-tuning (such as setting up judicious bounds on $\vb$ and carefully selecting the observation noise which must be bounded away from zero), the GP optimizer is \emph{unable} to find a reasonable fit as far as the Greeks are concerned. Instead, path-based design causes the GP surrogate to generate unstable and strongly oscillatory $\widehat{\Delta}$ and $\widehat{\Theta}$, making them practically unusable. This outcome is almost unavoidable for low $N$, but also manifests itself even with several hundred training points. Overall, we need to more than double the training set size in order to make path-based experimental design comparable to a space-filling one. Moreover, with an irregular path-based design, the GP model has a difficulty distinguishing signal from noise. Thus, increasing $\check{N}$ has only minor effect on learning Delta, instead the GP surrogate consistently overestimates the noise. This over-smoothes the data and removes most benefit of more precise inputs (higher $\check{N}$ in the experiment). 

Table  \ref{tbl:n-paths} in the Appendix contains the full summary statistics as we vary the design size. Table~\ref{tbl:n-paths} considers two different sampling frequencies in time which translate into different rectangular shapes for the training $\cD$. We observe a clear trade-off in the quality of $\widehat{\Delta}$ versus quality of $\widehat{\Theta}$: if we have more paths and lower sampling in time then the Delta estimation is better and Theta is worse. Conversely, training on fewer paths but with more frequent sampling in $t$ has adverse effect on $\widehat{\Delta}$. This pattern is intuitive for a data-driven method where quality of the approximation is explicitly linked to how much relevant information is provided in the training set. Other things being equal, we conclude that to learn Delta it is essential to have longer history rather than higher-frequency data.

\subsection{Results for the Local Volatility Model}\label{sec:lv}

To illustrate path-based training we take up the local volatility (LV) case study, where we train on an irregular grid obtained by generating 25 trajectories of $(S_t)$, saved at frequency $\Delta t= 0.04$, for a total of 250 training $(t^n,S^n_t)$ pairs. Figure \ref{fig:lv-greeks} shows the resulting Delta, Theta and Gamma approximatiors across three GP kernel families. As in the BS case study, the SE kernel has much too narrow credible bands, while the M32 kernel yields bands that are too wide. Unlike the first study, where SE-based model overcame the poor uncertainty quantification to yield the lowest RIMSE, here the SE kernel has clear trouble in providing a good fit, see~the significant error in estimating all three Greeks at both edges, especially for $S \gg 60$. This is confirmed by Table~\ref{tbl:kernel-lv} which shows that the SE kernel gives the worst fit among the three. We highlight the very high NLPD and very low coverage (i.e.~dramatic underestimation of posterior variance). The M52 and M32 kernels perform similarly for Delta, but M52 clearly outperforms both for Theta (where the credible band of the M32 model is absurdly wide) and for Gamma (where M32 is unstable, as expected). Table~\ref{tbl:n-lv} in the Appendix shows the impact of design size $N$ on the approximation quality. Overall, we thus again find Mat\'ern-5/2 to be the most appropriate kernel family.

For assessing Delta hedging, because we do not have the exact Delta instead of reporting \eqref{eq:muE}-\eqref{eq:VE} in Tables~\ref{tbl:kernel-lv}-\ref{tbl:n-lv} we report the variance of terminal hedging loss $E_T = W_T - \Phi(S_T)$. Lower $E_T$ indicates better hedging; in Table~\ref{tbl:kernel-lv} this is achieved with a M32 kernel. We note that in this case study, the approximation variance overestimates the impact on  hedging variance because there {is} a positive correlation between surrogate squared error $(\widehat{\Delta}-\Delta)^2$ (which is largest far from the strike $K$) and the specific form of $\sigma(t,S_t)$ which is also largest away from $K$. As a result, in the context of Remark~\ref{rem:varE} we obtain $\mathrm{Var}(E_T) < \mathrm{Var}( E^{(d)}_T ) + V_E$.

\begin{figure}[!ht]
\begin{tabular}{ccc}
 \includegraphics[width=0.32\textwidth]{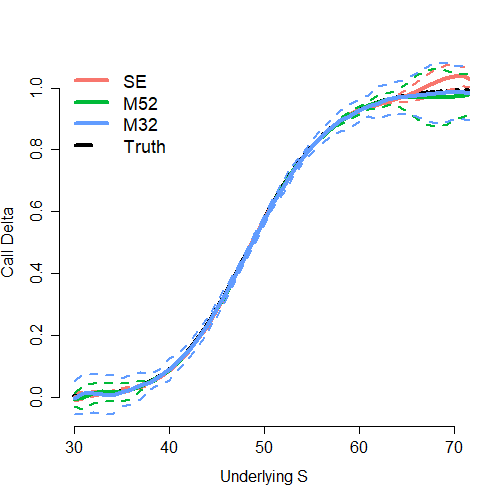} &
  \includegraphics[width=0.32\textwidth]{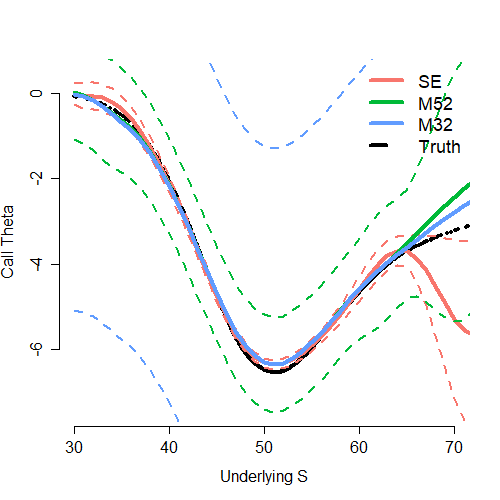} &
   \includegraphics[width=0.32\textwidth]{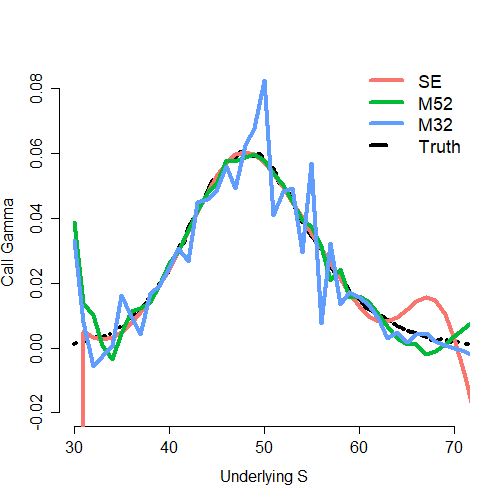} \\
   Delta & Theta & Gamma \\
    \includegraphics[width=0.32\textwidth]{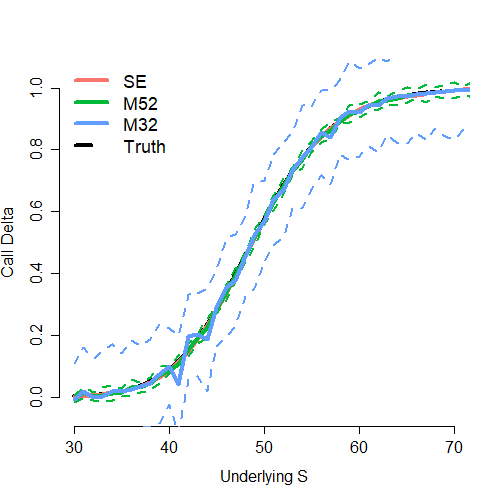} &
  \includegraphics[width=0.32\textwidth]{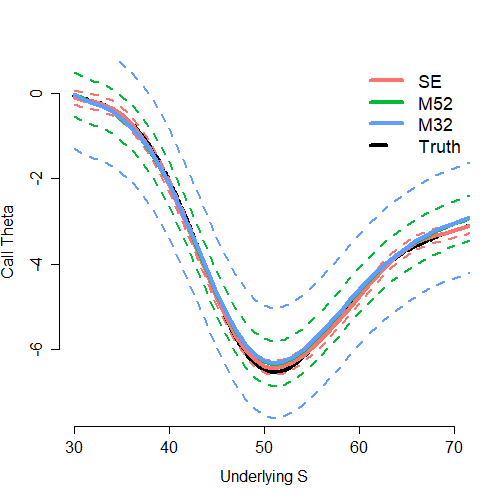} &
   \includegraphics[width=0.32\textwidth]{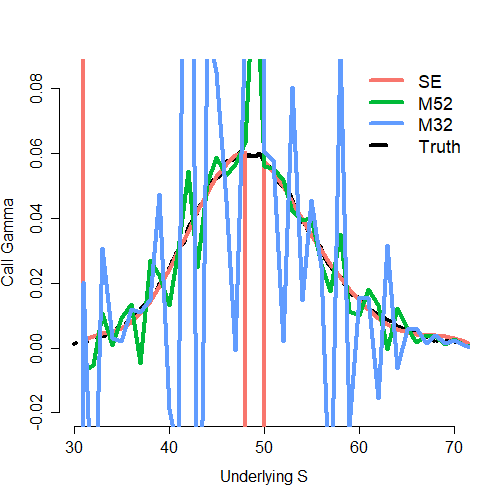} \\
   \end{tabular}
\caption{ \label{fig:lv-greeks} Fitted Greeks vs benchmark Greeks for the local volatility Call case study and three different kernel families and $t=0.08$. All models trained with 250 path-based inputs plus 50 virtual training points and include a linear trend function and constant estimated observation noise $\sigma^2_\epsilon$.}
\end{figure}

\begin{table}[ht]
	\caption{Effect of GP kernel family on learning the Delta in a local volatility model. We report 6  metrics for $\widehat{\Delta}$, as well as the RIMSE for $\Theta$ and option price $P$ (last 2 columns, cf.~\eqref{eq:rimse}). All metrics  are based on a gridded test  $\cD'$ of $31 \cdot 11 = 341$ sites, $\{ S_0: 29.4, 31.03, \ldots, 78.4\} \times \{t : 0, 0.04, \ldots, 0.36, 0.4\}$. Training set $\cD$ is of size 250 + 50 virtual points.
	\label{tbl:kernel-lv}}
	\begin{center} \vspace*{-0.2in}
				$$\begin{array}{lrrrrrrrrr}
					\hline
\text{Kernel} & \text{RIMSE} &  \text{MAD} & 95\% \text{Cvr} &  \text{Bias} & \text{NLPD} & Var(E_T) & \widehat\Theta_{Err} & \widehat{P}_{Err}  \\					\hline
\text{SE}  & 0.0400 & 0.0042 & 0.6246 & 0.00238 & 124.201 &1.316 & 1.382 & 0.080\\
\text{M52} & 0.0283 & 0.0018 & 0.9003 & 0.00250 & 79.095 & 0.915 & 0.905 & 0.046\\
\text{M32} & 0.0293 & 0.0021 & 0.9677 & 0.00288 & -2.248 & 0.718 & 0.858 & 0.038\\
\hline
				\end{array}$$\vspace*{-0.2in}
	\end{center}
\end{table}

We next use this LV case study to test further variations of the GP surrogates that are concerned with (i) role of the virtual training points; (ii) learning the observation noise; (iii) checking alternative GP regression tools. To do so, we construct several alternative GP models with results reported in Table~\ref{tbl:alt-lv}.  Our base case is a training set based on 20 paths (200 inputs), reinforced with 50 virtual points (20 deep ITM, 20 deep OTM, 10 at maturity) for a total training size of $|\cD| = 250$. The base GP uses a linear mean function $m(\bx) = \beta_0 + \beta_1 S$, a Mat\'ern-5/2 kernel, and a constant observation noise that is fitted via MLE. Henceforth, it is labeled as model M1. We then consider the following variants:

\begin{itemize}
  \item M2: same setup but with no virtual points at all (training set of size 200).
  \item M3: same setup, but only with 30 virtual points (10 deep ITM, 10 deep OTM, 10 at maturity). The alternatives M2/M3 test the impact of virtual points, namely using fewer of them relative to the base M1.

  \item M4: uses the given location-dependent observation noise $\hat{\sigma}(\bx_n)$ from the MC samples instead of a constant $\sigma_\epsilon$.

  \item M5: constant mean function $m(\bx) = \beta_0$ only.

  \item M6: pre-specified de-trending using a reference Black-Scholes model. Specifically, we de-trend by subtracting a Black-Scholes Call price based on a constant $\sigma=0.3$,
   utilizing the known maturity and spot. The GP surrogate is then fit to the ``residual''. M5/M6 illustrate the impact of the trend $m(\cdot)$ on the results. $m(\cdot)$ affects the hyperparameters of the surrogate and consequently has (an ambiguous) indirect effect on approximation quality.

  \item M7: {\tt hetGP} solver that non-parametrically learns non-constant observation noise $\sigma^2(\cdot)$ based on \cite{BGL18} and the corresponding \texttt{hetGP} package in \texttt{R}. The alternatives M4 and M7 test the role of observation noise. M4 replaces constant model-based observation noise $\sigma_\epsilon$ with a user-specified one; M7 nests M1 by using a more sophisticated GP approach.

\end{itemize}

\begin{table}[ht]
	\caption{Alternative GP surrogates for the local volatility case study. See main body for definitions of M1-M7. We report 7  metrics for $\widehat{\Delta}$, as well as the RIMSE for $\Theta$ and option price $P$. All metrics  are based on a test set  of 341 gridded sites, $\{ S_0: 29.4, 31.03, \ldots, 78.4\} \times \{t : 0, 0.04, \ldots, 0.36, 0.4\}$.
	\label{tbl:alt-lv}}
	\begin{center}\vspace*{-0.2in}
		{ $$\begin{array}{llrrrrrrrr}
					\hline
\text{Model} & \text{RIMSE} &  \text{MAD} & 95\% \text{Cvr} &  \text{Bias} & \text{NLPD} & \mu_E & V_E & \widehat\Theta_{Err} & \widehat{P}_{Err}  \\ \hline
M1 &0.0274 & 0.0024 & 0.9501 & 0.00018 & 62.26 & 0.0005 & 0.681 & 0.870 & 0.048\\
M2 &0.2465 & 0.0183 & 0.5982 & 0.08722 & 2.37 & 0.0395 & 19.578 & 3.621 & 1.611\\
M3 &0.0377 & 0.0043 & 0.9501 & -0.00257 & 60.86 & 0.0004 & 0.702 & 1.189 & 0.146\\
M4 &0.0371 & 0.0097 & 0.9677 & 0.00025 & -3.91 & 0.0004 & 1.075 & 1.327 & 0.093\\
M5 &0.0318 & 0.0038 & 0.9531 & 0.00019 & 47.56 & 0.0005 & 0.699 & 1.117 & 0.099\\
M6 &0.0067 & 0.0021 & 0.8944 & 0.00010 & -2.89 & 0.0000 & 0.029 & 1.497 & 0.026\\
M7 &0.0294 & 0.0027 & 0.9531 & 0.00014 & 52.59 & 0.0004 & 0.694 & 0.896 & 0.075\\
%M8 &0.0602 & 0.0051 & 0.9853 & 0.00063 & NaN & 0.0006 & 0.779 & 1.721 & 0.235\\
 \hline
\end{array}$$}
\end{center}\vspace*{-0.2in}
\end{table}

The following observations can be made regarding Table~\ref{tbl:alt-lv}. First, the addition of virtual points has a very strong positive effect. Without them (case M2), the surrogate performs very poorly. Thus, this is a ``zero-order'' feature of our approach. Moreover, the model strongly benefits from having plenty of virtual points (M3 vs M1) which are necessary to enforce the 0/1 gradient of the price surface at the edges of the domain in the asset coordinate. Second, specifying state-dependent observation noise degrades performance by introducing high-order
fluctuations into the surrogate. Similarly, a more sophisticated GP method targeting heteroskedasticity is not beneficial; there is no observed gain from adding complexity and the simpler base model wins out (M1 vs M4 or M7). Third, we observe that there are gains from having a reasonable trend function, in particular to capture the dominant trend in the asset coordinate. Such de-trending helps with spatial stationarity that GPs rely on. Thus, M5, which uses $m(x) = \beta_0$, performs worse than M1, while M6, which provides a highly accurate de-trending, helps the fit.

\subsection{Pathwise Hedging and Online Training}\label{sec:updating}

The left panel of Figure~\ref{fig:pnl-hist} illustrates using $\widehat{\Delta}$ to carry out Delta hedging along a  sample trajectory of $(S_t)$ as would be done in practice. We consider the local volatility case study;  in this scenario $S_0 = 44.70$ and $S_T = 41.66$, so the Call ends up OTM and terminal payoff and Delta are zero. We plot the benchmark $\Delta(t_k, S_{t_k})$ (red circles) and the GP-based $\widehat{\Delta}(t_k S_{t_k})$ (blue diamonds) along the 10 time-steps $t_k = k\Delta t$ with $\Delta t= 0.04$. We note that at the latter stages we have $S_t \simeq 35$ where the GP approximation is not so good (confirmed by the wide credible band of $\widehat{\Delta})$, however this has little effect on the hedging strategy since by that point Delta is almost zero anyway. On this particular path, we start with initial wealth of $W_0 = P(0,S_0) = 1.418$ and end up with the benchmark wealth of $W_T = E_T = -0.078$ (this error is driven by discrete hedging periods) and GP-based error of $\widehat{E}_T = -0.006$, i.e.~a difference of about 7 cents, in particular the GP strategy coming ahead.

\begin{figure}[ht]
  			\centering
  \includegraphics[height=2in,trim=0.1in 0.3in 0.1in 0.35in]{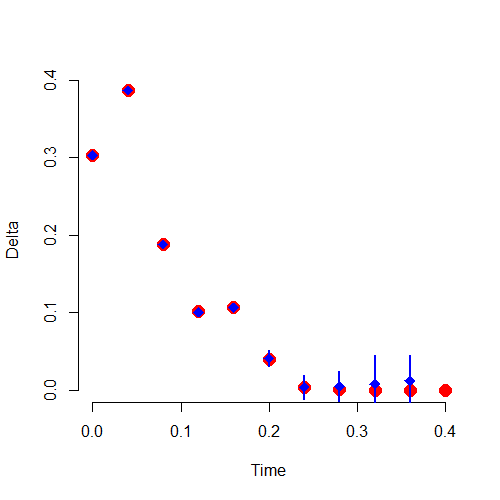}
   \includegraphics[height=2in,trim=0.1in 0.2in 0.1in 0.2in]{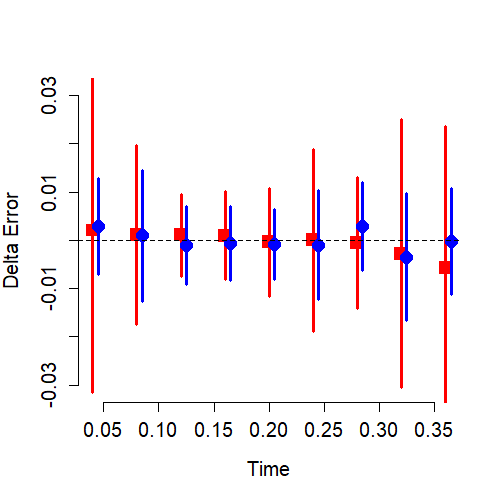}
\caption{ \label{fig:pnl-hist} 
\emph{Left:} a sample path showing Delta hedging in the local volatility model with 10 discretization periods ($\Delta t = 0.04$). Red circles indicate the benchmark Delta;  the blue vertical lines (resp.~blue diamonds) indicate the 95\% posterior bands (resp.~posterior mean) of the estimated GP Delta $\widehat{\Delta}(t_k,S_{t_k})$ based on 200+50 training inputs. \emph{Right:} Illustrating online learning of Delta along a high-frequency sampled price path. We plot the estimation error (relative to the ground truth $\Delta$) of the original $\widehat{\Delta}(t_k, S_{t_k})$ (in red) and of the recursively updated $\widehat{\Delta}^{online}(t_k,S_{t_k})$ (in blue), along with the respective 95\% credible intervals.
}
\end{figure}

\begin{remark}
  The outputted uncertainty quantification \eqref{eq:ci} for $\widehat{\Delta}$ can be used to implement a ``sticky'' hedge, where portfolio rebalancing is done only if there is a substantial trade needed, so as to save on transaction costs. Specifically, one could assume that rebalancing is carried out only when the old hedging position is outside the credible band $CI_{\alpha}$ of $\widehat{\Delta}(t_k, S_{t_k})$. In Figure~\ref{fig:pnl-hist}, this would imply no trading in the last 4 periods ($t > 0.24$), where $\widehat{\Delta}(t_k, S_{t_k}) \simeq 0$.
\end{remark}

To aid in such Delta-hedging along a path,
GP models are amenable to \emph{fast updating} in the context of augmenting with new data. Namely, the matrix form of the GP predictive equations \eqref{eq:gp_mean} can be exploited to facilitate adding new observations to improve the fit.  At the initial stage, the GP surrogate is trained on $N$ historical stock paths. Then one wishes to Delta hedge ``in real-time'' along a new $(S_t)$-trajectory. To do so, we sequentially collect $(k\Delta t, S_{k\Delta t}, P_{k\Delta t})$ values at regular intervals and then simultaneously estimate the ``in-sample'' $\widehat{\Delta}(k\Delta t, S_{k\Delta t})$ in order to find the new amount of shares to hedge with. In other words, at each hedging time instance we augment our training with the just-observed data and immediately estimate the Delta at the latest $(t,S_t)$ values. Such dynamic hedging mimics the online calibration that practitioners often carry out and amounts to recursively updating the original GP surrogate.

Adding a new training point $(\bx_{n+1},y_{n+1})$ to an existing GP model corresponds to augmenting the kernel matrix $\mathbf{K}$ with an extra row/column and analogously augmenting the other terms in the GP predictive equations. This can be done very efficiently through the so-called \emph{rank-1 update} if the GP hyperparameters are kept fixed, and requires just ${\cal O}(N^2)$ effort compared to ${\cal O}(N^3)$ effort to invert the full covariance matrix $\mathbf{K}$ in \eqref{eq:gp_mean}.

The right panel of Figure \ref{fig:pnl-hist} illustrates dynamic hedging through the above GP surrogate updating. We start with 20 historical paths sampled at $\Delta t = 0.04$ and then sequentially augment with high-frequency real-time trajectory sampled at $\Delta t = 0.004$ (reflecting the idea that the trader is now closely monitoring the option compared to originally downloading a fixed dataset). In Figure~\ref{fig:pnl-hist} we compare the initial $\widehat{\Delta}$ based on the 20 original paths versus the ``online'' $\widehat{\Delta}$, demonstrating  how the quality of the fit improves thanks to data fusion. Online learning of the Delta makes the estimation errors smaller (closer to zero in the figure) and furthermore narrows the posterior credible bands, hence doubly improving model fit: lower bias and higher credibility.

\begin{remark} One can of course proceed by brute force by simply re-estimating the entire GP surrogate as more data becomes available. That will likely give a slightly better fit. In comparison, online updating is more elegant conceptually and moreover is lightning fast since we do not need to keep re-running the MLE optimizer for the hyperparameters $\vb$.
\end{remark}

\subsection{Extending to Real-life Options Data}

Our method is directly applicable to dealing with observed option data since it requires no calibration beyond fitting the GP surrogate and is predicated on training using option prices, a quantity that is readily available in real life. To do so, one would switch to time-to-maturity $\tau$ parametrization, using historical data about options that already expired to generate a training set in the $(\tau,S)$ coordinates. Nevertheless, multiple challenges must be addressed before operationalizing this idea.

First, one must decide what does an ``option price'' mean, distinguishing between quotes, executed transactions and the issue of associated non-synchronous time stamps (e.g.~a market close price might not actually be a price that is directly relevant at any given fixed time of day). Moreover, quoted prices have bid/ask spreads which could be viewed as upper/lower bounds for $P(t,S)$. A related issue is the traded volume/open interest which could be interpreted as a proxy for quote quality.

There are several ways to match these features with the GPR setting:
\begin{itemize}
  \item Take $\sigma(\bx)$ to be proportional to the bid/ask spread (probabilistically ensuring that the fitted $\widehat{P}$ is within the spread)

  \item Take $\sigma(\bx)$ to be a function of Traded Volume/Order Imbalance to ensure that price of more liquid options are given more weight;

  \item Modify the Gaussian likelihood in \eqref{eq:noise} to account for the bid/ask spread. For example the \texttt{GPML Matlab} package implements a Beta likelihood that is appropriate for ``range regression''.

  \item Use a cut-off criterion to separate liquid contracts (where prices have to be matched either exactly or within bid/ask) and illiquid ones, where observations are treated only as ``vague'' suggestions.

 \end{itemize}
We remark that taking non-constant $\sigma(\bx)$ is statistically equivalent to a weighted least-squares criterion, i.e.~penalizing fitting errors more (resp.~less) when $\sigma(\bx)$ is small (resp.~large).

Second, one would have to contend with the irregular time series of financial data, with gaps due to weekends, holidays, missing data, etc. Of note, GPR is perfectly suited for that purpose since it does not assume or require any specific shape of the training set. At the same time, as demonstrated above in the context of irregular grid in the $S$-coordinate, irregular shapes can materially worsen the quality of the GP surrogate and its Greek estimators.

Finally, the described procedure so far assumed that time-to-maturity $\tau$ and asset price $S$ are sufficient statistics for determining the option price $P$.
For historical data, such as SPX options, we do observe strong time dependence that can be termed ``VIX effects'': for essentially same $(\tau, S)$ pairs the historical prices will be quite different (i.e.~different implied vol) on different days, indicating the presence of a further latent factor. As a first step, one would need to include calendar time as another covariate, working with the triple $(t,S,T)$ as postulated in a local volatility model. Another way to handle temporal non-stationarity would be to use a weighted regression, putting more weight on more recent data and discounting old data, which might minimize model mis-specification. A more complex extension would be to directly input VIX or other (stochastic volatility) factors when fitting the surrogate.

\section{Conclusion and Open Problems}

To conclude, we presented a framework of constructing GP surrogates for the purpose of learning option price sensitivities. Our method is completely statistical and fully generic, requiring simply a training set of (noisy) option prices. The GP surrogate is able to simultaneously provide estimates of Delta and Theta, along with their rigorously defined posterior uncertainty. Our case studies suggest that it is important to pick an appropriate kernel family, with the Mat\'ern-5/2 striking the best compromise across the numerous performance metrics we considered. {A GP M52 approximation offers a twice-differentiable surrogate for the option price that is smooth enough for Greek computation and flexible enough to capture the price surface.} Our analysis further highlights the importance of boundary conditions (specifically the gains provided by including virtual training points) and careful noise modeling (in particular letting the algorithm estimate observation variance). Another striking feature we observed is the significant impact of training set shape on quality of the Greeks approximation, including the benefit of space-filling.

An open problem is how to handle the
several well-known no-arbitrage constraints for the option price and its sensitivities. For example, a Call price must be convex monotone increasing in $S$ ($\Delta \ge 0,\Gamma \ge 0$), with slope less than unity ($\Delta \le 1$). It is also monotone decreasing in $t$,  $\widehat{\Theta} \le 0$.  To incorporate such features into a GP surrogate, one may consider monotonic GPs (see e.g.~\cite{riihimaki2010gaussian}) who make use of virtual GP-gradient observations, or finite-dimensional shape-constrained GPs \citep{DixonCrepey21}. Extending our \texttt{R} implementation to cover these is left for future research.
%The latter is available in the \texttt{MATLAB GPstuff} package but is not immediately compatible with our own \texttt{R} implementation and left for future research.
Another related work on incorporating gradient observations into a GP model is by \cite{chen2013enhancing}.

{A different comparator to the GP methodology are neural networks. In this framework, one runs a neural network (NN) regression to build a surrogate for the option price and then applies auto-differentiation to get the Greeks, see e.g.~\cite[Ch 2]{ChataignerThesis}. The latter step is available as a native function call for any NN architecture (i.e.~no analytic derivations necessary) in modern machine learning suites such as \texttt{TensorFlow}. Based on our preliminary experiments, NN-based Greeks tend to be unstable for small training sets ($N \ll 500$) as considered here, but perform very well for $N \ge 1000$. Full investigation of NN Greek approximators and respective uncertainty quantification for Delta hedging is left to future research.}

\appendix

\section{Proofs}\label{app:approximated_delta}

\begin{proof}[Proof of Proposition \ref{lemma:delta}]
Under the physical measure, we are assuming
$$dS_t = \mu S_t dt + \sigma(t,S_t)S_t dB_t,$$
where $B$ is a Brownian motion. We denote the price of a vanilla derivative with maturity $T$ by $P(t,S)$.
%The pricing PDE in this case is given by
%$$\Theta(t,S) + rS \Delta(t,S) + \frac{1}{2}\sigma^2(t,S)S^2 \Gamma(t,S) - rP(t,S) = 0.$$
The Delta hedging strategy perfectly replicates the derivative and can be described as
\begin{align}\label{eq:replication}
dP(t,S_t) = \Delta(t,S_t)dS_t + r(P(t,S_t) - \Delta(t,S_t)S_t)dt.
\end{align}

Let us consider an approximated Delta $\widehat{\Delta}$. The hedging error in continuous time follows the dynamics
\begin{align*}
dE_t = dP(t,S_t) - \widehat{\Delta}(t,S_t)dS_t - r(P(t,S_t) - \widehat{\Delta}(t,S_t)S_t)dt,
\end{align*}
with $E(0)=0$. {By the Delta-hedging replication Equation (\ref{eq:replication})}, we find
\begin{align*}
&dE_t =dP(t,S_t) - \widehat{\Delta}(t,S_t)dS_t - r(P(t,S_t) - \widehat{\Delta}(t,S_t)S_t)dt\\
&= \Delta(t,S_t)dS_t + r(P(t,S_t) - \Delta(t,S_t)S_t)dt \\
&- \widehat{\Delta}(t,S_t)dS_t - r(P(t,S_t) -  \widehat{\Delta}(t,S_t)S_t)dt\\
&= (\Delta(t,S_t) - \widehat{\Delta}(t,S_t)) (dS_t - rS_t dt)\\
&= (\Delta(t,S_t) - \widehat{\Delta}(t,S_t)) dX_t,
\end{align*}
where
$$dX_t = dS_t - rS_t dt = (\mu - r) S_t dt + \sigma(t,S_t)S_t dB_t.$$
Then
\begin{align*}
E_T = \int_0^T (\Delta(t,S_t) -\widehat{\Delta}(t,S_t)) dX_t.
\end{align*}
Under discrete-time delta hedging, we have
$$\widehat{\Delta}(t,S) = \sum_{k=0}^{K-1} \widehat{\Delta}(t_k,S_{t_k}) 1_{[t_k,t_{k+1})}(t)$$
and we find
\begin{align*}
E_T = \sum_{k=0}^{K-1} \int_{t_k}^{t_{k+1}} (\Delta(t,S_t) -\widehat{\Delta}(t_k,S_{t_k})) dX_t.
\end{align*}

Adding and subtracting $\Delta(t_k,S_{t_k})$ yields the result.
\end{proof}

\begin{proof}[Proof of Corollary \ref{cor:moments}]
The result follows from conditioning on $S_{t_k}$ and using the first two moments of $X$.
\end{proof}

\section{Additional Tables}

\begin{table}[ht]
	\caption{Effect of training set size on estimated Delta in a Black--Scholes model with learning based on $S$-paths. We report 8  metrics for $\widehat{\Delta}$, as well as the RIMSE for $\Theta$ and option price $P$ (last 2 columns, cf.~\eqref{eq:rimse}). All metrics  are based on a gridded test set of 1600 sites, $\{ S_0: 30, 30.5, \ldots, 69.5\} \times \{t : -0.01, 0.01, \ldots, 0.37\}$. GP model with Mat\'ern-5/2 kernel, linear trend function and constant estimated $\sigma_\epsilon$; all designs augmented with 50 additional virtual training points.
	\label{tbl:n-paths}}
	\begin{center}\vspace*{-0.2in}
		$$\begin{array}{lrrrrrrrrrr}
					\hline \hline
					N & \text{RIMSE} &  \text{MAD} & 95\% \text{Cvr} &  \text{Bias} & \text{NLPD} & {Var(E_T)} & \mu_E & V_E & \widehat\Theta_{Err} & \widehat{P}_{Err}  \\ \hline
\multicolumn{11}{c}{ \text{Paths with } \Delta t=0.04 \text{:  $N/10$ training paths}} \\ \hline
100 & 0.0663 & 0.0239 & 0.9762 & 0.0115 & -4.28 & 0.827 & 0.0168 & 0.604 & 2.958 & 0.433\\
150 & 0.0680 & 0.0347 & 0.9356 & 0.0138 & -4.35 & 0.744 &0.0104 & 0.542 & 2.228 & 0.387\\
200 & 0.0626 & 0.0293 & 0.9331 & 0.0107 & -4.50 & 0.653 & 0.0055 & 0.461 & 2.384 & 0.353\\
250 & 0.0602 & 0.0308 & 0.9281 & 0.0117 & -4.55 & 0.616 & 0.0055 & 0.415 & 2.148 & 0.327\\
300 & 0.0588 & 0.0284 & 0.9331 & 0.0100 & -4.60 & 0.592 & 0.0034 & 0.394 & 2.142 & 0.309\\
350 & 0.0554 & 0.0243 & 0.9381 & 0.0116 & -4.72 & 0.499 & -0.0011 & 0.292 & 2.010 & 0.286\\
400 & 0.0518 & 0.0213 & 0.9381 & 0.0108 & -4.87 & 0.460 & -0.0007 & 0.244 & 1.851 & 0.239\\
450 & 0.0501 & 0.0245 & 0.9525 & 0.0134 & -4.93 & 0.464 & 0.0009 & 0.246 & 1.993 & 0.261\\
%500 &0.0507 & 0.0228 & 0.9275 & 0.0121 & -4.91 & 0.479 & 0.0032 & 0.258 & 1.880 & 0.233\\
\hline
\multicolumn{11}{c}{\text{Paths with } \Delta t=0.02 \text{: $N/20$ training paths}} \\ \hline
100 &0.1348 & 0.0376 & 0.9419 & 0.0062 & -3.14 &  0.629 &  0.0096 & 0.353 & 1.314 & 0.491\\
150 &0.0786 & 0.0410 & 0.9525 & 0.0083 & -4.05 &  0.507 & 0.0103 & 0.242 & 1.370 & 0.339\\
200 &0.0857 & 0.0379 & 0.9219 & 0.0058 & -3.95 &  0.529  & 0.0127 & 0.264 & 1.551 & 0.333\\
250 &0.1141 & 0.0388 & 0.8619 & 0.0050 & -3.38 &  0.693 & 0.0145 & 0.419 & 1.191 & 0.389\\
300 &0.0856 & 0.0287 & 0.8750 & 0.0056 & -4.09 &  0.519  & 0.0109 & 0.264 & 1.083 & 0.334\\
350 &0.0571 & 0.0314 & 0.9419 & 0.0058 & -4.74 &  0.449 & 0.0069 & 0.205 & 1.223 & 0.235\\
400 &0.0586 & 0.0311 & 0.9312 & 0.0118 & -4.69 &  0.442 & 0.0069 & 0.193 & 1.213 & 0.271\\
450 &0.0569 & 0.0302 & 0.9194 & 0.0120 & -4.73 &  0.428 & 0.0061 & 0.182 & 1.225 & 0.267\\ \hline
%500 &0.0689 & 0.0295 & 0.8844 & 0.0127 & -4.33 & 0.450 & 0.0054 & 0.202 & 1.182 & 0.284\\ \hline
\end{array}$$
\end{center}\vspace*{-0.2in}
\end{table}

\begin{table}[htb]
	\caption{Effect of training set size $N=|\cD|$ on learning the Delta in the local volatility case study. We report 6  metrics for $\widehat{\Delta}$, as well as the RIMSE for $\Theta$ and option price $P$ (last 2 columns, cf.~\eqref{eq:rimse}). All metrics  are based on a gridded test set $\cD'$
of $31 \cdot 11 = 341$ sites, $\{ S_0: 29.4, 31.03, \ldots, 78.4\} \times \{t : 0, 0.04, \ldots, 0.36, 0.4\}$. The GP model uses  Mat\'ern-5/2 kernel, a linear trend function and estimated constant $\sigma^2_\epsilon$. 
	\label{tbl:n-lv}}
	\begin{center}\vspace*{-0.2in}
		$$\begin{array}{lrrrrrrrrr}
					\hline
					N & \text{RIMSE} &  \text{MAD} & 95\% \text{Cvr} &  \text{Bias} & \text{NLPD} & Var(E_T) & \widehat\Theta_{Err} & \widehat{P}_{Err}  \\ \hline
80 & 0.0375 & 0.0085 & 0.9589 & 0.00203 & 0.414 & 1.258  & 1.104 & 0.099\\
120 &0.0304 & 0.0052 & 0.9619 & 0.00200 & 9.752 & 1.063 & 1.043 & 0.088\\
160 &0.0290 & 0.0060 & 0.9560 & 0.00201 & 44.797 & 1.068 & 0.989 & 0.087\\
200 &0.0279 & 0.0028 & 0.9501 & 0.00213 & 62.243 & 1.024 & 0.909 & 0.082\\
240 &0.0282 & 0.0019 & 0.9032 & 0.00247 & 76.653 & 0.916 & 0.914 & 0.047\\
280 &0.0282 & 0.0018 & 0.8886 & 0.00249 & 100.673 & 0.918 &  0.895 & 0.047\\
320 &0.0283 & 0.0018 & 0.8651 & 0.00245 & 116.404 & 0.904 &  0.884 & 0.047\\
360 &0.0283 & 0.0017 & 0.8534 & 0.00264 & 113.659 & 0.959 &  0.898 & 0.044\\
400 &0.0282 & 0.0012 & 0.8182 & 0.00264 & 152.908 & 0.967 &  0.895 & 0.045\\ \hline
\end{array}$$
\end{center}
\end{table}

\bibliography{references}
\bibliographystyle{chicago}

\clearpage


\section*{Supplementary material}

# KrigHedge Demo

Mike Ludkovski

10/19/2020

This short RMarkdown file presents an illustrative use of Gaussian Process surrogates for estimation of option sensitivities. We directly embed `R` code snippets to showcase the straightforward use of the methodology.

## Training Dataset

```r
set.seed(101)
dx <- 0.01
xTest <- seq(27,73,by=1); tTest = 0.1;     # test set to be plotted
r <- 0.04  # interest rate
```

We consider learning the Greeks of a Call option within a Black–Scholes model. The Call has strike $K = 50$ and maturity $T = 0.4$. To do so, we employ a training set of 450 total training locations, with 400 actual inputs plus another 50 "virtual" inputs to capture the boundary conditions. Our task is to learn the Delta/Theta/Gamma of a Call as a function of current stock price $S_t$ (henceforth the spot) and time $t$. The inputs themselves are in the rectangle $[30, 70] \times [-0.01, 0.38]$.

The 400 training input-output tuples are constructed by sampling 400 locations via the space-filling Halton sequence (available in `randtoolbox` package) and then running Monte Carlo approximation of the respective option price through a plain Monte Carlo draw of 2500 i.i.d. samples based on the log-normal law of $S(T)$ (the simulation engine is viewed as a black-box for the modeler).

```r
N_tr <- 450
hltn <- halton(400, dim=2)
simDsgn <- cbind(t=c(-0.01+.38*hltn[,1],rep(0,50)),
                 spot=c(30+40*hltn[,2],rep(50,50)),
                 price=rep(0,N_tr), noise=rep(0,N_tr) )

for (i in 1:(N_tr-50)) {
  mcPrice <- BS_mc(2500,K=50,r=0.04, sigma=0.22,
                   T=0.4-simDsgn[i,"t"], S0=simDsgn[i,"spot"])
  simDsgn[i,"price"] <- mcPrice$mean
  simDsgn[i,"noise"] <- mcPrice$sd
}
```

The next snippet creates 50 additional "virtual" training points at the edges, namely 20 deep in-the-money ($S \in \{71, 72\}$), 20 deep out-of-the-money $S \in \{28, 29\}$) and 10 at maturity to capture the final payoff shape. The Figure below visualizes the overall training set, including the 400 MC-based inputs (in black) and the 50 virtual training points (in red).

```r
simDsgn[(N_tr-39):(N_tr),"t"] <- rep( seq(0,0.36,len=10), 4)
simDsgn[(N_tr-49):(N_tr),"noise"] <- 0
simDsgn[(N_tr-9):(N_tr),"price"] <- 71-exp(-r*(0.4-simDsgn[(N_tr-9):(N_tr),"t"]))*50
simDsgn[(N_tr-9):(N_tr),"spot"] <- 71
```

```r
simDsgn[(N_tr-19):(N_tr-10),"price"] <- 72-exp(-r*(0.4-simDsgn[(N_tr-19):(N_tr-10),"t"]))*50
simDsgn[(N_tr-19):(N_tr-10),"spot"] <- 72

simDsgn[(N_tr-29):(N_tr-20),"price"] <- 0
simDsgn[(N_tr-29):(N_tr-20),"spot"] <- 28

simDsgn[(N_tr-39):(N_tr-30),"price"] <- 0
simDsgn[(N_tr-39):(N_tr-30),"spot"] <- 29

simDsgn[(N_tr-49):(N_tr-40),"spot"] <- seq(32,68,len=10)
simDsgn[(N_tr-49):(N_tr-40),"price"] <- pmax(0, simDsgn[(N_tr-49):(N_tr-40),"spot"]-50)
simDsgn[(N_tr-49):(N_tr-40),"t"] <- 0.4
```

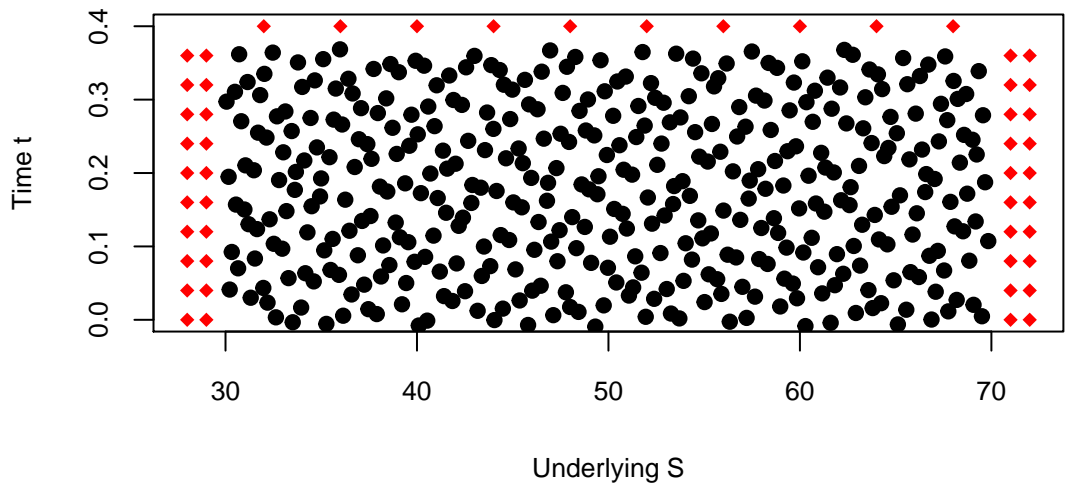

## Training the GP surrogate

With the training set of approximate option prices constructed, we are ready to train a GP surrogate. Below we employ the `DiceKriging` package and select the Matern-52 kernel family, linear trend function, estimated constant observation noise (`nugget`) and genetic-algorithm optimizer for maximum likelihood estimation of the GP hyperparameters.

```r
gpModel <- km(  formula = y ~ 1 + spot,  # linear trend function
                design =simDsgn[,1:2],  response=simDsgn[,"price"],
                nugget.estim=TRUE,  # lear
                # alternatively: use the estimated simulation noise
                #noise.var=pmax(1e-7, simDsgn[,"noise"]),
                covtype="matern5_2",   # can also try "gauss" or "matern3_2"
                optim.method="gen",
                estim.method = "MLE",
                lower=c(0.1,8), upper=c(2,50),  # bounds on lengthscales
                # the "control" parameters below handle speed versus risk of
                # converging to local minima.  See "rgenoud" package for details
```

2

```
                  control=list(max.generations=100,pop.size=100,
                               wait.generations=8,
                               solution.tolerance=1e-5,
                               maxit = 1000, trace=F
                  ))
print(coef(gpModel))
```

```
## $trend1
## [1] -15.58528
##
## $trend2
## [1] 0.5031343
##
## $range
## [1]  0.8441864 15.3521758
##
## $shape
## numeric(0)
##
## $sd2
## [1] 9.081037
##
## $nugget
## [1] 0.006124068
```

### Getting the Greeks

The next code generates the posterior mean/variance of the gradients of `gpModel`. For the Price this is just a `predict` command. For Delta and Theta, we utilize the helper `gpDerivative` function that implements formulas (2.14)-(2.15) in the paper for the Matern-5/2 kernel.

```
testSet <- data.frame(t=rep(tTest,length(xTest)),spot=xTest)
# Delta: gradient with respect to x_2
DeltaTest <- gpDerivative(fit=gpModel, testSet,i=2)

# Theta -- gradient with respect to x_1
ThetaTest <- gpDerivative(fit=gpModel,testSet,i=1)

# Price itself
PriceTest <-  predict(gpModel,newdata=testSet,type="UK")
```

To compute second order sensitivity, specifically option Gamma, we employ finite differences:

$$\frac{\partial^2 P}{\partial S^2}(t,S) \simeq \frac{P(t,S+h) - 2P(t,S) + P(t,S-h)}{h^2}.$$

```
GammaTest <- xTest   # evaluate one predictive site at a time
for (jj in 1:length(xTest)) {
  GammaTriple <- data.frame(spot=c(xTest[jj]-dx,xTest[jj],xTest[jj]+dx),
                                   t=c(tTest,tTest,tTest))
  triple <-  predict(gpModel,newdata=GammaTriple,type="UK")$mean
  GammaTest[jj] <- (triple[3]-2*triple[2]+triple[1])/dx^2   # Gamma approximation
}
```

We can now plot the Greeks and their posterior uncertainty (credible bands). This is equivalent to the plots in Figure 4.

```r
plot(xTest,BScall(t=tTest,T=0.4,S=xTest,K=50,r=0.04,q=0,sigma=0.22,isPut=0)$Theta,
    col="black",type="l", lwd=4, lty=2,
    xlab='Underlying S', ylab='Call Theta', bty='n', ylim=c(-6.2,0.7))
  lines(xTest,ThetaTest$m, lwd=4, col="orange")
  lines(xTest,ThetaTest$m+1.96*ThetaTest$covmat, col="orange",lwd=2,lty=2)
  lines(xTest,ThetaTest$m-1.96*ThetaTest$covmat, col="orange",lwd=2,lty=2)
```

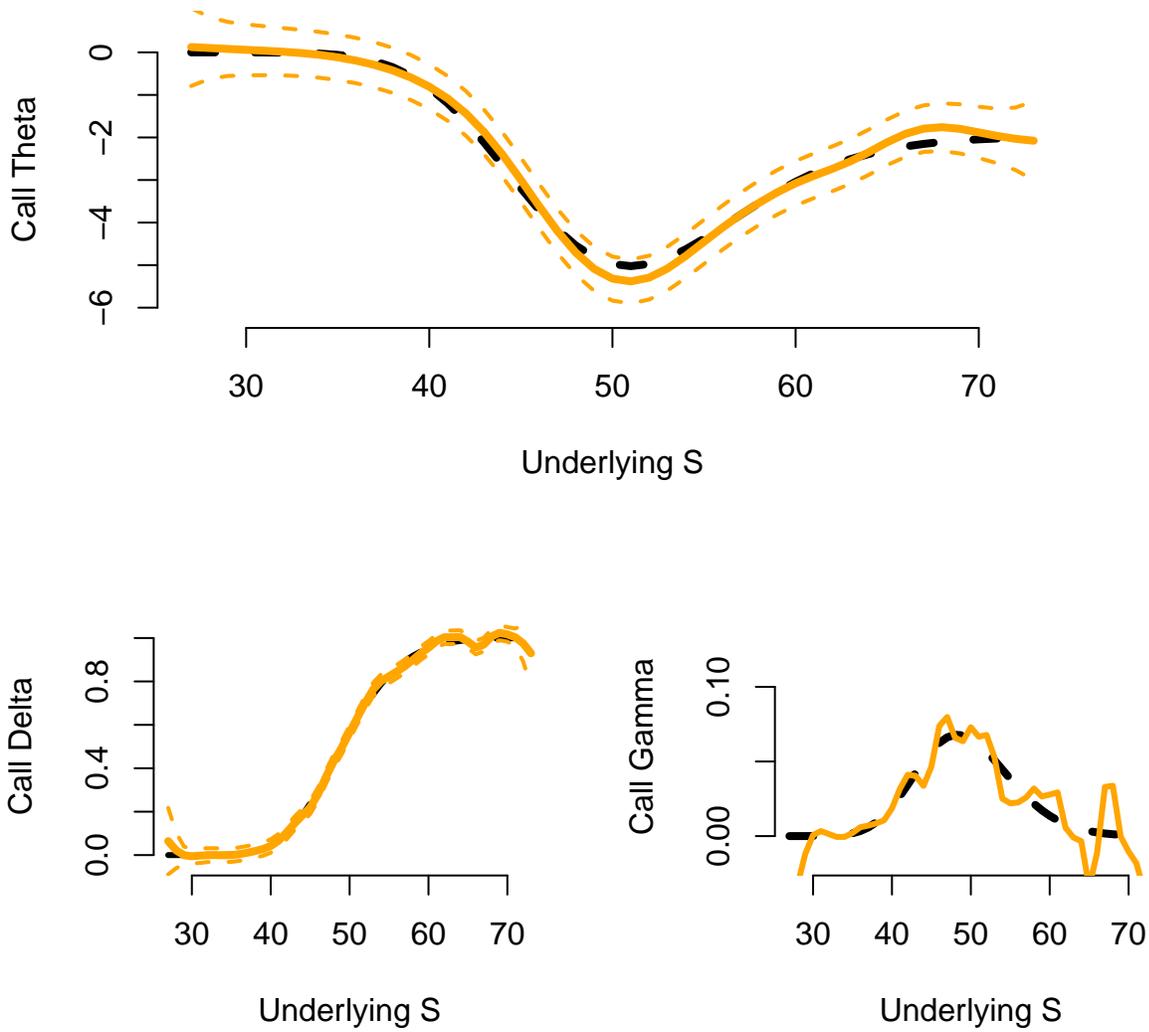



\end{document}